\documentclass[12pt]{amsart}
\usepackage{amssymb,color}
\numberwithin{equation}{section} 

\usepackage{graphicx}

\newcommand{\bea}{\begin{eqnarray}}
\newcommand{\eea}{\end{eqnarray}}
\newcommand{\ba}{\begin{array}}
\newcommand{\ea}{\end{array}}
\newcommand{\edc}{\end{document}}
\newcommand{\bc}{\begin{center}}
\newcommand{\ec}{\end{center}}
\newcommand{\be}{\begin{equation}}
\newcommand{\ee}{\end{equation}}

\newcommand{\dsf}{\displaystyle\frac}

\def\cb{{\mathcal B}}

\def\ce{{\mathcal E}}

\def\cg{{\mathcal G}}

\def\bc{{\mathbb C}}
\def\bbf{{\mathbb F}}

\def\bn{{\mathbb N}}

\def\bq{{\mathbb Q}}
\def\br{{\mathbb R}}

\def\bz{{\mathbb Z}}

\def\a{\alpha}
\def\b{\beta}
\def\g{\gamma}  \def\G{\Gamma}
\def\d{\delta}  

\def\e{\epsilon}

\def\m{\mu}

\def\r{\rho}
\def\s{\sigma} 
\def\t{\theta}

\def\w{\omega} \def\Om{\Omega}

\def\h{{\mathbf{h}}}

\def\xb{{\mathbf{x}}}
\def\sb{{\mathbf{s}}}

\newtheorem{thm}{Theorem}[section]
\newtheorem{lem}[thm]{Lemma}

\newtheorem{prop}[thm]{Proposition}

\theoremstyle{remark}
\newtheorem{rem}{Remark}[section]

\begin{document}

\title[On $p$-adic Gibbs measures]
{Phase transitions for $P$-adic Potts model on the Cayley tree of
order three}


\author{Farrukh Mukhamedov}
\address{Farrukh Mukhamedov\\
 Department of Computational \& Theoretical Sciences\\
Faculty of Science, International Islamic University Malaysia\\
P.O. Box, 141, 25710, Kuantan\\
Pahang, Malaysia} \email{{\tt far75m@yandex.ru} {\tt
farrukh\_m@iium.edu.my}}

\author{Hasan Ak\i n}
\address{Hasan Ak\i n, Department of Mathematics, Faculty of Education,
 Zirve University, Kizilhisar Campus, Gaziantep, 27260, Turkey}
\email{{\tt hasanakin69@gmail.com}}

\begin{abstract}
In the present paper, we study a phase transition problem
for the $q$-state $p$-adic Potts model over the
Cayley tree of order three. We consider a more general notion
of $p$-adic Gibbs measure which depends on parameter $\rho\in\bq_p$.
Such a measure
is called {\it generalized $p$-adic quasi Gibbs measure}. When $\rho$ equals to $p$-adic
exponent, then it coincides with the $p$-adic Gibbs measure. When $\rho=p$, then
it coincides with $p$-adic quasi Gibbs measure. Therefore, we investigate
two regimes with respect to the value of $|\rho|_p$. Namely, in the first
regime, one takes $\rho=\exp_p(J)$ for some $J\in\bq_p$, in the second one $|\rho|_p<1$.
In each regime, we first find conditions for the existence of generalized
$p$-adic quasi Gibbs measures. Furthermore, in the first regime, we
establish the existence of
the phase transition under some conditions. In the second regime,
when $|\r|_p,|q|_p\leq p^{-2}$ we prove the existence of a quasi phase transition. It
turns out that if
$|\r|_p<|q-1|_p^2<1$ and $\sqrt{-3}\in\bq_p$, then one finds the existence of
the strong phase transition.

\end{abstract}

\maketitle

\section{introduction}

It is know that (see \cite{W}) the $q$-state Potts model
is one of the most studied models in
statistical mechanics. It has wide theoretical interest and practical
applications. Originally, the Potts model was introduced as a
generalization of the Ising model to more than two spin components.
The model has enough rich structure to illustrate
almost every conceivable nuance of statistical mechanics. In
\cite{Per1,Per2,Ga} the phase diagrams of the model on the Bethe
lattices (Cayley tree in our terminology)
were studied and the pure phases of the ferromagnetic
Potts model were found. Note that the Bethe lattices were fruitfully used,
providing a deeper insight into the behavior of the Potts models.

We have to stress that one of the central problems
in the theory of Gibbs measures of lattice systems
is to describe infinite-volume (or limiting) Gibbs measures corresponding to a given
Hamiltonian. A complete analysis of this set is often a difficult problem. Many papers
have been devoted to these studies when the underlying lattice is a Cayley tree (see for
example, \cite{Ga,G,GR,R,Za}).

On the other hand, since the 1980s various models in physics described in the
language of $p$-adic analysis (see \cite{ADV,ADFV,FO,MP,V1,V2}), and
numerous applications of
such an analysis to mathematical physics have been studied in
\cite{ABK,Kh1,Kh2,Ko,MP,Vi,VVZ}. In those studies, it
appeared unconventional probability measures.
This means that a number of $p$-adic models in physics cannot be described
using ordinary Kolmogorov's probability theory. New probability
models, namely $p$-adic ones were investigated in
\cite{BD,K3,KYR,KL,Lu,Ro,F1}. Note that such kind of
probabilities measures have many applications in algebraic number theory (see for example,
\cite{BP,BS,Ch,HT,Vi}.

In \cite{GMR1,MR1,MR2} we have developed a
$p$-adic probability theory approaches to study $p$-adic Potts models
on Cayley tree of order two. We especially interested in the
construction of $p$-adic Gibbs measures for the mentioned model. It was
proved that if $q$ is divisible by $p$ then a phase transitions
occurs for the $q$-state Potts model. Further, in \cite{Mq} we have
introduced a
new kind of $p$-adic measures, associated with $q+1$-state Potts model,
called {\it $p$-adic quasi Gibbs measure}, which is totaly different
from the $p$-adic Gibbs measure.
 Note that such measures present more natural concrete examples of
$p$-adic Markov processes (see \cite{F31}).  We established the
existence $p$-adic quasi Gibbs measures for the mentioned model on a Cayley tree of order two.
Moreover, if  $q$ is divisible by $p$, then we prove
the occurrence of a strong phase transition.
If $q$ and $p$ are relatively prime, then there is a quasi phase transition.
These results totaly different from
the results of \cite{MR1,MR2}.
It is natural to ask: do we have the similar phase transitions if one decreases the order
of the tree?

In the present paper, our aim is to investigate a phase transition problem
for the $q$-state $p$-adic Potts model over the
Cayley tree of order three. In this paper, we consider a more general notion
of $p$-adic Gibbs measure which depends on parameter $\rho\in\bq_p$ (see \cite{Mq2}).
Such a measure
is called {\it generalized $p$-adic quasi Gibbs measure}. When $\rho$ equals to $p$-adic
exponent, then it coincides with the $p$-adic Gibbs measure. When $\rho=p$, then
it coincides with $p$-adic quasi Gibbs measure. Therefore, in the sequel, we will
consider two regimes with respect to the values of $|\rho|_p$. Namely, in the first
regime, one takes $\rho=\exp_p(J)$ for some $J\in\bq_p$, in the second one $|\rho|_p<1$.
In each regime, we first find conditions for the existence of generalized
$p$-adic quasi Gibbs measures. Furthermore, we will study the occurrence of phase transitions.
In the first regime, in \cite{MR1} it was predicted that a phase transition
may occur if $|q|_p<1$ (i.e. $q$ is divisible by $p$). In this paper,
we able to establish the existence of
the phase transition under a stronger condition, i.e. when $|q|_p\leq p^{-3}$.
If $|q|_p=1$, such a
case was not investigated neither in \cite{MR1} nor \cite{Mq}. So, in that case, we
proved the existence of a phase transition under some conditions. In the second regime,
when $|\r|_p,|q|_p\leq p^{-2}$ we will prove the existence of a quasi phase transition. It
turns out that if
$|\r|_p<|q-1|_p^2<1$, then one finds the existence of
the strong phase transition.

\section{Preliminaries}

\subsection{$p$-adic numbers}

In what follows $p$ will be a fixed prime number, and by $\bq_p$ it is denoted the field of
$p-$adic numbers, which is a completion of the rational numbers
$\bq$ with respect to the norm $|\cdot|_p:\bq\to\br$
given by
\begin{eqnarray}
|x|_p=\left\{
\begin{array}{c}
  p^{-r} \ x\neq 0,\\
  0,\ \quad x=0,
\end{array}
\right.
\end{eqnarray}
here, $x=p^r\frac{m}{n}$ with $r,m\in\bz,$ $n\in\bn$,
$(m,p)=(n,p)=1$. A number $r$ is called \textit{a $p-$order} of $x$
and it is denoted by $ord_p(x)=r.$ The
absolute value $|\cdot|_p$, is non- Archimedean, meaning that it
satisfies the ultrametric triangle inequality $|x + y|_p \leq
\max\{|x|_p, |y|_p\}$.

Any $p$-adic number $x\in\bq_p$, $x\neq 0$ can be uniquely represented in the form
\begin{equation}\label{canonic}
x=p^{ord_p(x)}(x_0+x_1p+x_2p^2+...),
\end{equation}
where $x_j$ are integers, $0\leq x_j\leq p-1$,
$x_0>0$, $j=0,1,2,\dots$ In this case $|x|_p=p^{-ord_p(x)}$.

We respectively denote the set of all {\it $p-$adic integers} and
{\it units} of $\bq_p$ by
$$\bz_p=\{x\in\bq_{p}: |x|_p\leq1\}, \quad \bz_p^{*}=\{x\in\bq_{p}: |x|_p=1\}.$$

Any nonzero $p-$adic number $x\in\bq_p$ has a unique representation
of the form $x =\cfrac{x^{*}}{|x|_p}$, where $x^{*}\in\bz_p^{*}$.

\begin{lem}[Hensel's Lemma, \cite{VVZ}]\label{Hensel}
Let $f(x)$ be polynomial whose the coefficients are $p-$adic
integers. Let $\theta$ be a $p-$adic integer such that for some
$i\geq 0$ we have
$$
f(\theta)\equiv 0 \ (mod \ p^{2i+1}),
$$
$$
f'(\theta)\equiv 0 \ (mod \ p^{i}), \quad f'(\theta)\not\equiv 0 \
(mod \ p^{i+1}).
$$
Then $f(x)$ has a unique $p-$adic integer root $x_0$ which satisfies
$x_0\equiv \theta\ (mod \ p^{i+1}).$
\end{lem}

Let $p$ be a prime number, $q\in\bn$, $a\in\bbf_p$ with $a\neq \bar{0}.$
The number $a$ is called \textit{a $q$-th power
residue modulo $p$} if the the following equation
\begin{eqnarray}\label{kthresidue}
x^q=a
\end{eqnarray} has a solution in $\bbf_p$.

\begin{prop}[\cite{Ros}]\label{aisresidueofp}
Let $p$ be an odd prime number, $q\in\bn$, $d=(q,p-1),$ and $a\in\bbf_p$ with $a\neq\bar{0}.$ Then the following statements hold true:
\begin{itemize}
  \item [(i)] $a$ is the $q$-th power residue modulo $p$ if and only if one has
$a^{\frac{p-1}{d}}=\bar{1};$
  \item [(ii)] If $a^{\frac{p-1}{d}}=\bar{1}$ then the equation \eqref{kthresidue} has $d$ number of solutions in $\bbf_p$.
\end{itemize}
\end{prop}

Let $B(a,r)=\{x\in \bq_p : |x-a|_p< r\}$, where $a\in \bq_p$, $r>0$.
The $p$-adic logarithm is defined by series
$$
\log_p(x)=\log_p(1+(x-1))=\sum_{n=1}^{\infty}(-1)^{n+1}\dsf{(x-1)^n}{n},
$$
which converges for every $x\in B(1,1)$. And $p$-adic exponential is
defined by
$$
\exp_p(x)=\sum_{n=1}^{\infty}\dsf{x^n}{n!},
$$
which converges for every $x\in B(0,p^{-1/(p-1)})$.

\begin{lem}\label{21} \cite{Ko},\cite{VVZ} Let $x\in
B(0,p^{-1/(p-1)})$ then we have $$ |\exp_p(x)|_p=1,\ \ \
|\exp_p(x)-1|_p=|x|_p<1, \ \ |\log_p(1+x)|_p=|x|_p<p^{-1/(p-1)} $$
and $$ \log_p(\exp_p(x))=x, \ \ \exp_p(\log_p(1+x))=1+x. $$
\end{lem}

In what follows we will use the following

\begin{lem}[\cite{KMM}]\label{pr} If $|a_i|_p\leq 1$, $|b_i|_p\leq 1$, $i=1,\dots,n$, then
\begin{equation*}
\bigg|\prod_{i=1}^{n}a_i-\prod_{i=1}^n b_i\bigg|_p\leq \max_{i\leq
i\leq n}\{|a_i-b_i|_p\}
\end{equation*}
\end{lem}

Denote
\begin{equation}\label{Exp}
\ce_p=\{x\in\bq_p: \ |x|_p=1, \ \ |x-1|_p<p^{-1/(p-1)}\}.
\end{equation}

So, from Lemma \ref{21} one concludes that if $x\in\ce_p$, then
there is an element $h\in B(0,p^{-1/(p-1)})$ such that
$x=\exp_p(h)$.

Note that the basics of $p$-adic analysis, $p$-adic mathematical physics
are explained in \cite{Ko,Ma,S,Ro,VVZ}.

\subsection{$p$-adic measure}

Let $(X,\cb)$ be a measurable space, where $\cb$ is an algebra of
subsets $X$. A function $\m:\cb\to \bq_p$ is said to be a {\it
$p$-adic measure} if for any $A_1,\dots,A_n\subset\cb$ such that
$A_i\cap A_j=\emptyset$ ($i\neq j$) the equality holds
$$
\mu\bigg(\bigcup_{j=1}^{n} A_j\bigg)=\sum_{j=1}^{n}\mu(A_j).
$$

A $p$-adic measure is called a {\it probability measure} if
$\mu(X)=1$.  One of the important condition (which was already
invented in the first Monna--Springer theory of non-Archimedean
integration \cite{Mona}) is boundedness, namely a $p$-adic
probability measure $\m$ is called {\it bounded} if $\sup\{|\m(A)|_p
: A\in \cb\}<\infty $. We should stress that the boundedness
condition by itself provides a fruitful
integration theory (see for example \cite{Kh07}). Note that, in
general, a $p$-adic probability measure need not be bounded
\cite{K3,KL,Ko}. For more detail information about $p$-adic measures
we refer to \cite{AKh},\cite{K3},\cite{KhN},\cite{Ro}.

\subsection{Cayley tree}

Let $\Gamma^k_+ = (V,L)$ be a semi-infinite Cayley tree of order
$k\geq 1$ with the root $x^0$ (whose each vertex has exactly $k+1$
edges, except for the root $x^0$, which has $k$ edges). Here $V$ is
the set of vertices and $L$ is the set of edges. The vertices $x$
and $y$ are called {\it nearest neighbors} and they are denoted by
$l=<x,y>$ if there exists an edge connecting them. A collection of
the pairs $<x,x_1>,\dots,<x_{d-1},y>$ is called a {\it path} from
the point $x$ to the point $y$. The distance $d(x,y), x,y\in V$, on
the Cayley tree, is the length of the shortest path from $x$ to $y$.

Recall a coordinate structure in $\G^k_+$:  every vertex $x$ (except
for $x^0$) of $\G^k_+$ has coordinates $(i_1,\dots,i_n)$, here
$i_m\in\{1,\dots,k\}$, $1\leq m\leq n$ and for the vertex $x^0$ we
put $(0)$.  Namely, the symbol $(0)$ constitutes level 0, and the
sites $(i_1,\dots,i_n)$ form level $n$ ( i.e. $d(x^0,x)=n$) of the
lattice.


Let us set
$$ W_n=\{x\in V| d(x,x^0)=n\}, \ \ \
V_n=\bigcup_{m=1}^n W_m, \ \ L_n=\{l=<x,y>\in L | x,y\in V_n\}.
$$
For $x\in \G^k_+$, let us denote
\begin{equation}\label{S(x)}
S(x)=\{y\in W_{n+1} :  d(x,y)=1 \} \ \ x\in W_n.
\end{equation} here
This set is called a set of {\it direct successors} of $x$.

Using the coordinate system one can define translations of $\G^k_+$
by
\begin{equation}\label{transl}
\tau_{(i_1,\dots,i_n)}(j_1,\dots,j_m)=(i_1,\dots,i_n,j_1,\dots,j_m).
\end{equation}

\subsection{Generalized $p$-adic quasi Gibbs
measure}

In this section we define a notion of generalized $p$-adic quasi
Gibbs measure in a general setting, i.e. for arbitrary
nearest-neighbor models (see \cite{Mq2}).

Let $\Phi=\{0,1,2,\cdots,q-1\}$, here $q\geq 1$, ($\Phi$ is called a
{\it state space}) and is assigned to the vertices of the tree
$\G^k_+=(V,\Lambda)$. A configuration $\s$ on $V$ is then defined as
a function $x\in V\to\s(x)\in\Phi$; in a similar manner one defines
configurations $\s_n$ and $\w$ on $V_n$ and $W_n$, respectively. The
set of all configurations on $V$ (resp. $V_n$, $W_n$) coincides with
$\Omega=\Phi^{V}$ (resp. $\Omega_{V_n}=\Phi^{V_n},\ \
\Omega_{W_n}=\Phi^{W_n}$). One can see that
$\Om_{V_n}=\Om_{V_{n-1}}\times\Om_{W_n}$. Using this, for given
configurations $\s_{n-1}\in\Om_{V_{n-1}}$ and $\w\in\Om_{W_{n}}$ we
define their concatenations  by
$$
(\s_{n-1}\vee\w)(x)= \left\{
\begin{array}{ll}
\s_{n-1}(x), \ \ \textrm{if} \ \  x\in V_{n-1},\\
\w(x), \ \ \ \ \ \ \textrm{if} \ \ x\in W_n.\\
\end{array}
\right.
$$
It is clear that $\s_{n-1}\vee\w\in \Om_{V_n}$.

Let $H_n$ be a Hamiltonian given by
\begin{equation}\label{Ham}
H_n(\s)=\sum_{<x,y>\in L_n}\Psi_{x,y}(\s),  \ \ \s\in\Om_{V_n}, \
n\in\mathbb{N},
\end{equation}
where $\Psi_{x,y}(\s):\Om_{V_n}\to\bz$ is a given function for every
$x,y$.

A construct of a  generalized $p$-adic quasi Gibbs measure
corresponding to the model is given below.

Assume that  $\h: V\setminus\{x^{(0)}\}\to\bq_p^{\Phi}$ is a
mapping, i.e. $\h_x=(h_{0,x},h_{1,x},\dots,h_{q-1,x})$, where
$h_{i,x}\in\bq_p$ ($i\in\Phi$) and $x\in V\setminus\{x^{(0)}\}$.
Given $n\in\bn$, we consider a $p$-adic probability measure
$\m^{(n)}_{\h,\rho}$ on $\Om_{V_n}$ defined by
\begin{equation}\label{mu}
\mu^{(n)}_{\h,\rho}(\s)=\frac{1}{Z_{n,\rho}^{(\h)}}{\rho}^{H_n(\s)}\prod_{x\in
W_n}h_{\s(x),x}
\end{equation}
Here, $\rho\in\bq_p$, $\s\in\Om_{V_n}$, and $Z_{n,\rho}^{(\h)}$ is
the corresponding normalizing factor called a {\it partition
function} given by
\begin{equation}\label{ZN1}
Z_{n,\rho}^{(\h)}=\sum_{\s\in\Omega_{V_n}}{\rho}^{H_n(\s)}\prod_{x\in
W_n}h_{\s(x),x}.
\end{equation}

Note that, in general, $Z_{n,\rho}^{(\h)}$ could be zero, but we
will show that, in a Potts model case, the partition function is not
zero. In this paper, we are interested in a construction of an
infinite volume distribution with given finite-dimensional
distributions. More exactly, we would like to find a
$p$-adic probability measure $\m$ on $\Om$ which is compatible
with given ones $\m_{\h,\rho}^{(n)}$, i.e.
\begin{equation}\label{CM}
\m(\s\in\Om: \s|_{V_n}=\s_n)=\m^{(n)}_{\h,\rho}(\s_n), \ \ \
\textrm{for all} \ \ \s_n\in\Om_{V_n}, \ n\in\bn.
\end{equation}

In general, \`{a} priori the existence such a kind of measure $\m$
is not known, since there is not much information on topological
properties, such as compactness, of the set of all $p$-adic measures
defined even on compact spaces \cite{F4}. Note that certain
properties of the set of $p$-adic measures has been studied in
\cite{kas2,kas3}, but those properties are not enough to prove the
existence of the limiting measure. Therefore, at a moment, we are going to employ
so called the $p$-adic Kolmogorov's extension Theorem (see
\cite{GMR},\cite{KL}) which is based on {\it compatibility
condition} for the measures $\m_{\h,\rho}^{(n)}$, $n\geq 1$, i.e.
\begin{equation}\label{comp}
\sum_{\w\in\Om_{W_n}}\m^{(n)}_{\h,\rho}(\s_{n-1}\vee\w)=\m^{(n-1)}_{\h,\rho}(\s_{n-1}),
\end{equation}
for any $\s_{n-1}\in\Om_{V_{n-1}}$. This condition according to the
theorem implies the existence of a unique $p$-adic measure
$\m_{\h,\rho}$ defined on $\Om$ with a required condition
\eqref{CM}. Such a measure $\m_{\h,\rho}$ is said to be {\it a
genaralized $p$-adic quasi Gibbs measure} corresponding to the
model. Note that more general theory of $p$-adic measures has been
developed in \cite{kas1}.

By $GQ\cg_\rho(H)$ we denote the set of all generalized $p$-adic
quasi Gibbs measures associated with functions $\h=\{\h_x,\ x\in
V\}$. If there are at least two distinct generalized $p$-adic quasi
Gibbs measures such that at least one of them is unbounded,
then we say that \textit{a phase transition}
occurs. If there are two
different functions $\sb$ and $\h$ defined on $\bn$ such that there
exist the corresponding measures $\m_{\sb,\rho}$, $\m_{\h,\rho}$,
and they are bounded, then we say there is a \textit{quasi phase
transition}. If one finds two different functions $\sb$
and $\h$ defined on $\bn$ such that there exist the corresponding
measures $\m_{\r,\sb}$ and $\m_{\r,\h}$, for which one is bounded, another one
is unbounded, and there is a sequence of sets $\{A_n\}$
such that $A_n\in\Om_{V_n}$ with $|\m_{\r,\s}(A_n)|_p\to 0$,
$|\m_{\r,\h}(A_n)|_p\to\infty$ as $n\to\infty$, then we say that there
occurs a \textit{strong phase transition} (see \cite{Mq2}).

\section{$p$-adic Potts model and its $p$-adic quasi Gibbs measures}

In this section we consider the $p$-adic Potts model where spin
takes values in the set $\Phi=\{0,1,2,\cdots,q-1\}$. In what follows, we will always
assume that $q\geq 3$. In this case,
for every $\s\in\Om_{V_n}$ one has
$\Psi_{x,y}(\s)=\delta_{\s(x),\s(y)}$, i.e.  the Hamiltonian of {\it
$q$-state Potts} model has the following form
\begin{equation}\label{Potts}
H_n(\s)=\sum_{<x,y>\in L_n}\delta_{\s(x),\s(y)},  \ \
\s\in\Om_{V_n}, \  n\in\mathbb{N},
\end{equation}
where $\delta$ is the Kronecker symbol. Note that when $q=1$, then
the corresponding model reduces to the $p$-adic Ising model. Such a
model was investigated in \cite{GMR,KM}.

Given $\rho\in\bq_p$ one can construct generalized $p$-adic quasi
Gibbs measures corresponding to the $q$-state Potts model. In what
follows, we are interested in the existence of phase transition for
the Potts model.

\begin{rem} Note that if one takes $\rho=p$, then the generalized
$p$-adic quasi Gibbs measure reduces to the {\it $p$-adic quasi
Gibbs measure} (see \cite{Mq}). If one takes $\rho\in\ce_p$ with
$h_x\in \ce_p$, then the defined measure reduces to {\it $p$-adic
Gibbs measure} (see \cite{MR1}).
\end{rem}

\begin{rem} Note that in \cite{GMR,KM,Mq,MR1,MR2} phase transitions
in class of  $p$-adic quasi Gibbs measures, $p$-adic Gibbs measures
for Ising and Potts models on Cayley tree have been studied over the
Cayley tree of order two. When a state space $\Phi$ is countable,
the corresponding $p$-adic Gibbs measures (resp. generalized $p$-adic quasi Gibbs measures)
have been investigated in
\cite{KMM,M} (resp. \cite{M1}).
\end{rem}

Using the same argument as \cite{MR1,KMM} we can prove the following

\begin{thm}\label{comp1} The measures $\m^{(n)}_{\h,\rho}$, $
n=1,2,\dots$ (see \eqref{mu}) associated with $q$-state Potts model
\eqref{Potts} satisfy the compatibility condition \eqref{comp} if
and only if for any $n\in \bn$ the following equation holds:
\begin{equation}\label{eq1}
\hat h_{x}=\prod_{y\in S(x)}{\mathbf{F}}(\hat \h_{y};\r),
\end{equation}
here and below a vector $\hat \h=(\hat h_1,\dots,\hat
h_{q-1})\in\bq_p^{q-1}$ is defined by a vector
$\h=(h_0,h_1,\dots,h_{q-1})\in\bq_p^{q}$ as follows
\begin{equation}\label{H}
\hat h_i=\frac{h_i}{h_0}, \ \ \ i=1,2,\dots,q-1
\end{equation}
and mapping ${\mathbf{F}}:\bq_p^{q-1}\times\bq_p\to\bq_p^{q-1}$ is
defined by
${\mathbf{F}}(\xb;\r)=(F_1(\xb;\r),\dots,F_{q-1}(\xb;\r))$ with
\begin{equation}\label{eq2}
F_i(\xb;\r)=\frac{(\r-1)x_i+\sum\limits_{j=1}^{q-1}x_j+1}
{\sum\limits_{j=1}^{q-1}x_j+\r}, \ \ \xb=\{x_i\}\in\bq_p^{q-1}, \ \
i=1,2,\dots,q-1.
\end{equation}
\end{thm}

\begin{rem}\label{r3} In what follows, without loss of generality, we may
assume that $h_0=1$. Otherwise, in \eqref{mu} we multiply and divide
the expression on the right hand side by $\prod_{x\in W_n}h_{0,x}$,
and after replacing $h_i$ by $h_i/h_0$, we get the desired equality.
\end{rem}

\section{Existence of generalized $p$-adic quasi Gibbs measures}

In this section we are going to establish the existence of
generalized $p$-adic quasi Gibbs measures for the Potts model
\eqref{Potts} on a Cayley tree of order 3, i.e. $k=3$.
In the sequel, for the sake of simplicity
we will always assume that $p>3$.

Recall that a function $\h=\{\h_x\}_{x\in V\setminus\{x^0\}}$ is
called {\it translation-invariant} if $\h_{\tau_x(y)}=\h_{y}$ for
all $x,y\in V\setminus\{x^0\}$. A $p$-adic measure $\m_\h$,
corresponding to a translation-invariant function $\h$, is called
{\it translation-invariant generalized $p$-adic quasi Gibbs
measure}. Note that the translational-invariance of
$\h=\{\h_x\}_{x\in V\setminus\{x^0\}}$ implies that $\h_x=\h_y$ for
all $x,y\in V\setminus\{x^0\}$.

 Let us first restrict ourselves to
the description of translation-invariant solutions of \eqref{eq1},
namely $\h_x=\h(=(h_0,h_1,\dots,h_q))$ for all $x\in V$. Then
\eqref{eq1} can be rewritten as follows
\begin{equation}\label{eq11}
\hat h_{i}=\bigg(\frac{(\theta-1)\hat h_{i}+\sum_{j=1}^{q}\hat
h_j+1} {\sum_{j=1}^{q}\hat h_j+\theta}\bigg)^3, \ \ i=1,2,\dots,q.
\end{equation}

One can see that $(\underbrace{1,\dots,1,h}_m,1,\dots,1)$ is an
invariant line for \eqref{eq11} ($m=1,\dots,q-1$). Over such a line
equation \eqref{eq11} reduces to the following fixed point problem
\begin{equation}\label{eq12}
x=f_\r(x),
\end{equation}
where
\begin{equation}\label{f(x)}
f_\r(x)=\bigg(\frac{\r x+q-1}{x+\r+q-2}\bigg)^3.
\end{equation}

A simple calculation shows that \eqref{eq12} has a form
$$
(x-1)(x^3-Ax^2-Bx+C)=0,
$$
where

\begin{eqnarray}\label{ABC}
\left\{
\begin{array}{lll}
A=\r^3-3\r-3q+5,\\[2mm]
B=\r(\r-3)^2+9q+3\r q(\r-2)-3q^2-7, \\[2mm]
C=(q-1)^3.
\end{array}
\right.
\end{eqnarray}

 Hence, $x_0=1$ solution defines a generalized $p$-adic quasi
Gibbs measure $\m_0$.

Now we are interested in finding other solutions of \eqref{eq12},
which means we need to solve the following one
\begin{equation}\label{eq13}
x^3-Ax^2-Bx+C=0.
\end{equation}

By performing standard substitution $x=z+A/3$, one can reduce the
equation \eqref{eq13} to
\begin{equation}\label{eq-alp1}
z^3+\a z=\b,
\end{equation}
where
\begin{equation}\label{alpha-beta1}
\a=-\bigg(B+\frac{A^2}{3}\bigg), \ \
\b=\frac{2}{3^3}A^3+\frac{AB}{3}-C.
\end{equation}

Now we are going to consider two regimes with respect to $\rho$.
Namely, (i) $\rho\in\ce_p$, and (ii) $|\rho|_p<1$.

\subsection{Regime $\r\in\ce_p$}. In this case, due to Lemma \ref{21} one has $|\r|_p=1$ and
$|\r-1|_p<1$.

Note that in this regime, the measure corresponding to the
solution $x=1$ is always $p$-adic Gibbs measure.

Direct calculations show that one has
\begin{eqnarray}\label{A}
A&=&3+3(1-\r)+(\r^3-1)-3q,\\[2mm] \label{alpha2}
 \alpha
&=&\frac{1}{3}(\r-1)^2\bigg(-9+9q-9(\r-1)+6q(\r-1)-9(\r-1)^2\nonumber \\[2mm]
&&-6(\r-1)^3-(\r-1)^4\bigg)\\[2mm] \label{beta2}
\beta&=&\frac{1}{27}(\r-1)^3\bigg(54-81q+27q^2+81(\r-1)-81q(\r-1)\nonumber \\[2mm]
&&+108(\r-1)^2-81q(\r-1)^2+81(\r-1)^3-18q(\r-1)^3\nonumber \\[2mm]
&&+54(\r-1)^4+18(\r-1)^5+2(\r-1)^6 \bigg)
\end{eqnarray}

In \cite{MR1} we have proved the following

\begin{thm}[\cite{MR1}]\label{Un} Let $|q|_p=1$ and $\r\in\ce_p$. Then there is a unique
$p$-adic Gibbs measure for the Potts model over the Cayley tree of
order $k$ ($k\geq 1$).
\end{thm}

This theorem implies that if $|q|_p=1$ then there is no $p$-adic Gibbs measure except for
$\m_0$. But further, we will show that under the mentioned condition, one can find
generalized $p$-adic quasi Gibbs measures.\\

Now we first assume that $|q|_p<1$ to find
nontrivial $p$-adic Gibbs measures. Under this condition, from the
strong triangle inequality from \eqref{alpha2},\eqref{beta2} one
finds $|\a|_p=|\r-1|_p^2$, and $|\b|_p=|\r-1|_p^3$.  We emphasize that the existence of
nontrivial $p$-adic Gibbs measure means that a solution of
\eqref{eq13} belongs to $\ce_p$, which implies that one can find $h\in\bq_p$
such that $x=\exp_p(h)$.

\begin{lem}\label{sol} Let $\r\in\ce_p$ and $|q|_p<1$. Then a solution $x$ of \eqref{eq13}
belongs to $\ce_p$ iff the corresponding solution of \eqref{eq-alp1} belongs to
$\bz_p\setminus\bz_p^*$.
\end{lem}

\begin{proof}
First denote $\g=3(1-\r)+(\r^3-1)-3q$. It is clear that $|\g|_p<1$.
Assume that $x\in\ce_p$ is a solution of
\eqref{eq13}. Then for the solution of \eqref{eq-alp1} with \eqref{A} we get
$$
|z|_p=\bigg|x-\frac{A}{3}\bigg|_p=\bigg|x-1+1-\frac{3+\g}{3}\bigg|_p=
|x-1-\g|_p<1.
$$
This means that a solution of \eqref{eq-alp1} must be in
$\bz_p\setminus\bz_p^*$.

Now assume that a solution $z$ of \eqref{eq-alp1} belongs to $\bz_p\setminus\bz_p^*$.
Then
$x=z+A/3$ is a solution of \eqref{eq13}. From \eqref{A}
we find that $|A|_p=1$, and one gets
$|x|_p=|z+A/3|_p=1$. Moreover, we have
$$
|x-1|_p=\bigg|z+\frac{A}{3}-1\bigg|_p=|z+\g|_p<1,
$$
hence $x\in\ce_p$.
\end{proof}

So, to find nontrivial $p$-adic Gibbs measures, we need to solve \eqref{eq-alp1} over the
set $\bz_p\setminus\bz_p^*$. Hence, assume that a solution $z$ of \eqref{eq-alp1}
has the following form $z=p^ky$,
where $y\in\bz^*_p$ and $k\geq 1$. Let us substitute the last form
to \eqref{eq-alp1} and one finds
\begin{equation}\label{eq-y}
y^3+ay=b
\end{equation}
where $a=p^{-2k}\a, \ b=p^{-3k}\b$. It is clear that
$|a|_p=(p^k|\r-1|_p)^2$, $|b|_p=(p^k|\r-1|_p)^3$. Now we chose $k$ such a
way that $p^k|\r-1|_p=1$. Note that such a number does exist.  So,
$|a|_p=|b|_p=1$.

Now we want to solve \eqref{eq-y} in $\bz_p^*$. To do it we are
going to apply Theorem \ref{criterionforp>3}.

Let us consider the canonical forms of $a$ and $b$, i.e.
$$
a=a_0+a_1p+\cdots, \ \ b=b_0+b_1p+\cdots,
$$
where $a_0,b_0\in\{1,\dots,p-1\}$. It then follows from
\eqref{alpha2} and \eqref{beta2} that
\begin{equation}\label{ab1}
a=p-3+(p-1)p+3q-3(\r-1)+\eta p^n, \ \ b=2-3q+3(\r-1)+\eta_1p^m, \ \
n,m\geq 2,
\end{equation}
where $\eta,\eta_1\in \bz_p$. So, $a_0=p-3$ and $b_0=2$. Hence,
$D_0=-4p(p^2-9p+27)$, where $D_0=-4a_0^3-27b_0^2$. This means
$D_0\equiv 0 \ (\textrm{mod}\ p)$. In this case, due to Theorem
\ref{numberofsolutionsforp>3} we need to compute $D=-4a^3-27b^2$. Let us
assume that $|q|_p\leq p^{-3},|\r-1|_p\leq p^{-3}$. Then from \eqref{ab1} one
finds that
$$
a^3\equiv (p-3)^3+3(p-3)^2(p-1)p+3(p-3)(p-1)^2p^2 \ (\textrm{mod}\
p^3), \ \ \ b^2\equiv 4 \ (\textrm{mod}\ p^3).
$$
hence, we have
$$
D=-\bigg(4p(p^2-9p+27)-4\cdot 3^3p+45p^2+\varepsilon p^3\bigg)=-p^2(9+\varepsilon p)
$$
where $\varepsilon\in\bz_p$. This means $d_0=9$ and $ord_p(D)$ is divisible by 2. Moreover,
$d_0^{p-2/2}\equiv 1 \ (\textrm{mod}\ p)$ holds. Then according to
Theorem \ref{numberofsolutionsforp>3}
we conclude that the equation \eqref{eq-y} has three distinct solutions in $\bz^*_p$.

Thus, the equation
\eqref{eq-alp1} has three distinct solutions belonging to $\bz_p\setminus\bz^*_p$.
Hence, by Lemma \ref{sol} the equation \eqref{eq13} has three solutions. Therefore,
by $\m_i$ ($i=1,2,3$) we denote the corresponding $p$-adic Gibbs measures.
Due to Theorem \ref{comp1} one can formulate the following

\begin{thm}\label{1-e} Let $p>3$, $|q|_p\leq p^{-3}$, $\r\in\ce_p$ with $|\r-1|_p\leq p^{-3}$. Then
for the $p$-adic Potts model \eqref{Potts} on the Cayley tree of
order three, there exist four translation-invariant $p$-adic Gibbs measures.
\end{thm}

Note that in \cite{MR1} if $|q|_p<1$ we did not have other extra conditions for the
existence of the $p$-adic Gibbs measures. From Theorem \ref{1-e} we see that for the existence of
the Gibbs measures one needs $|q|_p<p^{-3}$ extra condition.\\

Now assume that $|q|_p=1$. In this case we are going to establish
the existence of a generalized $p$-adic quasi Gibbs measure. In the
sequel, for the sake of simplicity we impose on $q$ the following constrains:
\begin{equation}\label{qq}
|q-1|_p=1, \ \ |2q-1|_p=1.
\end{equation}
Then from \eqref{A},\eqref{alpha2} and \eqref{beta2} we immediately
find that $|A|_p=1$, $|\a|_p=|\r-1|_p^2$ and $|\b|_p=|\r-1|_p^3$. In this
case, the equation \eqref{eq-alp1} has no solution belonging to
$\bz_p^*$. Indeed, assume that $z\in\bz_p^*$ is a solution, then we
have $|\a|_p=|z\a|_p=|\b-z^3|_p=1$, but this contradicts to
$|\a|_p<1$. Consequently, any solution of \eqref{eq-alp1} belongs to
$\bz_p\setminus\bz_p^*$.  Now using the same argument as above we
reduce the equation \eqref{eq-alp1} to \eqref{eq-y} with
$y\in\bz_p^*$, $a=p^{-2k}\a, \ b=p^{-3k}\b$.  By the same argument, we choose $k$ such a way
that
$p^k|\r-1|_p=1$. So,
$|a|_p=|b|_p=1$.

Assume that $|\r-1|_p\leq p^{-3}$. Then from \eqref{alpha2}, \eqref{beta2} one gets
\begin{equation}\label{ab2}
a\equiv 3(q-1)\ (\textrm{mod}\
p^3), \ \ \ b\equiv 2(q-1)(2q-1) \ (\textrm{mod}\ p^3).
\end{equation}

Let us compute $D=-4a^3-27b^2$. Then from \eqref{ab2} one
finds that
\begin{eqnarray*}
D&=&-\bigg(4\cdot 27(q-1)^3+27\cdot 4(q-1)^2(2q-1)^2+\eta |\r-1|\bigg)\\[2mm]
&=&-108q(q-1)^2(4q-3)-\eta |\r-1|_p, \ \ \textrm{for some} \ \ \eta\in\bz_p.
\end{eqnarray*}
Hence, if $|\r-1|_p<|4q-3|_p$ we conclude that
$|D|_p=|4q-3|_p$.

Now taking into account Theorem \ref{numberofsolutionsforp>3} we can formulate the following

\begin{thm}\label{1-2e} Let $p>3$, $\r\in\ce_p$ with $|\r-1|_p\leq p^{-3}$,
$|\r-1|_p<|4q-3|_p$ and $|q|_p=|q-1|_p=|2q-1|_p=1$. Let
$u_{n+3}=b_0u_{n}-a_0u_{n+1}$ with $u_1=0$, $u_2=-a_0$, $u_3=-b_0$, where
$a_0 \equiv 3(q-1)\ (\textrm{mod}\ p)$, $b_0\equiv 2(q-1)(2q-1)\
(\textrm{mod}\ p)$ with  $a_0,b_0\in\{1,2,\dots,p-1\}$.
Then
for the $q$-state $p$-adic Potts model \eqref{Potts} on the Cayley
tree of order three, there exist one translation-invariant $p$-adic
Gibbs measure and
\begin{enumerate}
\item[(i)] three translation-invariant generalized
$p$-adic quasi Gibbs measures if one of the following conditions are satisfied:
\begin{itemize}
\item[(a)] $|4q-3|_p=1$ and $u_{p-2}\equiv 0 \ (\textrm{mod} \ p)$;

\item[(b)] $|4q-3|_p=p^{-2k}$ for some $k\geq 1$, and $-108q(q-1)^2(4q-3)p^{-2k}$ is a quadratic
residue modulo $p$.
\end{itemize}

\item[(ii)] one translation-invariant generalized
$p$-adic quasi Gibbs measure if one of the following conditions are satisfied:
\begin{itemize}
\item[(c)] $D_0u_{p-2}^2\not\equiv 0\ \textrm{and}\  9a_0^2 \ (\textrm{mod} \ p)$;

\item[(d)] $|4q-3|_p=p^{-(2k+1)}$, \ $k\geq 1$;

\item[(e)] $|4q-3|_p=p^{-2k}$, and $-108q(q-1)^2(4q-3)p^{-2k}$ is not a quadratic
residue modulo $p$.
\end{itemize}

\item[(iii)] otherwise, there is not translation-invariant generalized
$p$-adic quasi Gibbs measure.
\end{enumerate}
\end{thm}

We should stress that in \cite{MR1} it was not investigated the case $|q|_p=1$. From the theorem
we conclude in the mentioned setting one can find new kind of Gibbs measures.

Let us consider some more concrete examples.

{\sc Example 1}. Assume that $p=5$ and $q=4$. It is clear that
$|q|_5=1$ and \eqref{qq} is satisfied. From \eqref{alpha2} and
\eqref{beta2} we get that $a_0=4$, $b_0=2$. One can check that
$D_0=-4\cdot 91$, $u_3=-2$. Hence, we have $D_0u_3^2\equiv 9\cdot
a_0^2 \ (\textrm{mod}\ 5)$, this with Theorem \ref{numberofsolutionsforp>3}
implies \eqref{eq-alp1} has no solution. In this case, for the model there is
only $p$-adic Gibbs measure.

{\sc Example 2}. Assume that $p=5$ and $q=7$. It is clear that
$|q|_5=1$ and \eqref{qq} with $|4q-3|_5=5^{-2}$ is satisfied. One can see that
$108q(q-1)^2(4q-3)5^{-2}\equiv 1 \ (\textrm{mod}\ 5)$. Since 1 is a quadratic residue
modulo 5, this yields that $108q(q-1)^2(4q-3)5^{-2}$ is so. Hence, due to Theorem \ref{1-2e}
there are three translation-invariant generalized $p$-adic quasi Gibbs measures.

\subsection{Regime $|\r|_p<1$} This regime is totaly different from
the previous case. Now we are interested in the existence of
generalized $p$-adic quasi Gibbs measures.

Direct calculations show that one has
\begin{eqnarray}\label{alpha3}
 \alpha
&=&\frac{1}{3}(\r-1)^2\bigg(-4+3q-5\r+6q\r+3\r^2-2\r^3-\r^4\bigg)\\[2mm]
\label{beta3}
\beta&=&\frac{1}{27}(\r-1)^3\bigg(38-63q+27q^2-30\r+27q\r+\nonumber \\[2mm]
&&+39\r^2-27q\r^2+5\r^3-18q\r^3-6\r^4+6\r^5+2\r^6\bigg)
\end{eqnarray}

Assume that $|q|_p<1$. Then from \eqref{alpha3},\eqref{beta3} one
finds that $|\a|_p=1$, and $|\b|_p=1$ if $p\neq 19$, and $|\b|_p<1$
if $p=19$.

Let us first consider $p\neq 19$. Furthermore, for the sake of simplicity, we will
assume that $|\r|_p,|q|_p\leq p^{-2}$.  Then one can
see that
\begin{equation}\label{aabb}
\a=-\frac{4}{3}+\e p^2, \ \ \b=-\frac{38}{27}+\e_1 p^2, \ \ \e,\e_1\in\bz_p.
\end{equation}
So, $a_0\equiv -4/3 \ (\textrm{mod}\ p)$, and $b_0\equiv
-38/27\ (\textrm{mod}\ p)$. Moreover, one can find that
$D=-(4\a^3+27\b^2)=-44+\eta p^2$, where $\eta\in\bz_p$.
Hence,
$$
|D|_p=
\left\{
\begin{array}{ll}
1, \ \ \textrm{if} \ p\neq 11,\\[2mm]
p^{-1}, \ \ \textrm{if} \ p=11.
\end{array}
\right.
$$
Therefore, due to Theorem
\ref{numberofsolutionsforp>3} we conclude that
the equation \eqref{eq-alp1} has a unique solution either $p=11$ or $p\neq 11$ and $D_0u^2_{p-2}\not\equiv 0; 9\
(\textrm{mod}\ p)$.

Now assume that $p=19$. Then $|\b|_p=p^{-1}, |\a|_p=1$.  One can see that
$a_0=5$. We know that $-5$ is not a quadratic residue modulo 19, therefore, due to
Theorem \ref{criterionforp>3} I.2 the equation \eqref{eq-alp1} has no solution in $\bz_p^*$.
Hence, the equation may have a solution in $\bz_p$. In this case, we have $z=p^ky$,
where $k\geq 1$ and
$y\in\bz_p^*$. Then the equation is reduced to \eqref{eq-y} with $a=p^{-2k}\a$ and $b=p^{-3k}\b$.
It is clear that $|a|_p=p^{2k}$, $|b|_p=p^{2k}$. Take $k=1$, then again
Theorem \ref{criterionforp>3} and Theorem \ref{numberofsolutionsforp>3}
I.4 imply that the equation \eqref{eq-y} has a unique solution.
Hence, the equation
\eqref{eq-alp1} has a unique solution belonging to $\bz_p\setminus\bz_p^*$.

So, one can formulate the
following

\begin{thm}\label{1-r} Let $p>3$. Assume that $|\r|_p,|q|_p\leq p^{-2}$. Then
for the $q$-state $p$-adic Potts model \eqref{Potts} on the Cayley
tree of order three, there exist
two translation-invariant
generalized $p$-adic quasi Gibbs measures
\begin{enumerate}
\item[(i)] if $p=11$;

\item[(ii)] if $p=19$;

\item[(iii)] if $p\neq 11,19$ and
$D_0u^2_{p-2}\not\equiv 0; 9\
(\textrm{mod}\ p)$, where $a_0\equiv -4/3 \ (\textrm{mod}\ p)$, and $b_0\equiv
-38/27\ (\textrm{mod}\ p)$, where $a_0,b_0\in\{1,2,\dots,p-1\}$.
\end{enumerate}
\end{thm}

Now we assume that $|q-1|_p<1$. In what follows, we will assume that $|\r|_p<|q-1|_p^2$.
Then from \eqref{alpha3},\eqref{beta3} one gets
\begin{eqnarray}\label{alpha4}
 \alpha
&=&\frac{1}{3}(\r-1)^2\bigg(-1+3(q-1)+\e\r\bigg)\\[2mm]
\label{beta4}
\beta&=&\frac{1}{27}(\r-1)^3\bigg(2-9(q-1)+27(q-1)^2+\e_1\r\bigg)
\end{eqnarray}
where $\e,\e_1\in\bz_p$. Hence, we rewrite \eqref{alpha4},\eqref{beta4} as follows
\begin{eqnarray}\label{alpha5}
 \alpha
=-\frac{1}{3}+(q-1)+\eta (q-1)^3, \ \ \ \beta=-\frac{2}{27}+\frac{q-1}{3}-(q-1)^2+\eta_1(q-1)^3
\end{eqnarray}
where $\eta,\eta_1\in\bz_p$. It is clear that $|\a|_p=|\b|_p=1$.
Then one finds that
\begin{eqnarray}\label{DD}
-D&=&4\a^3+27\b^3=-\frac{4}{27}+\frac{4}{3}(q-1)-4(q-1)^2\nonumber\\[2mm]
&&+\frac{4}{27}+3(q-1)^2-\frac{4}{3}(q-1)+4(q-1)^2+\tilde\eta(q-1)^3\nonumber\\[2mm]
&=&(q-1)^2(3+\tilde\eta (q-1))
\end{eqnarray}
where $\tilde\eta\in\bz_p$. This means $|D|_p=|q-1|_p^2$ and $d_0=-3$. Hence, Theorem
\ref{numberofsolutionsforp>3} implies that the equation \eqref{eq-alp1} has three solutions if
$\sqrt{-3}$ does exists. Otherwise, it has a unique solution.
Therefore, we have the following

\begin{thm}\label{2-r} Let $p>3$. Assume that $|\r|_p<|q-1|^2_p<1$. Then
for the $q$-state $p$-adic Potts model \eqref{Potts} on the Cayley
tree of order three, there exist
\begin{enumerate}
\item[(i)] four translation-invariant
generalized $p$-adic quasi Gibbs measures if $\sqrt{-3}$ exists;

\item[(ii)] two translation-invariant
generalized $p$-adic quasi Gibbs measures if $\sqrt{-3}$ does not exists.
\end{enumerate}
\end{thm}

\begin{rem}\label{41} Let $|\r|_p<|q-1|^2_p<1$ be satisfied, and we assume that
the equation \eqref{eq13} has three solutions $x_1,x_2,x_3$, which means $\sqrt{-3}\in\bq_p$.
Due to Vieta's
Theorem one has
\begin{equation}\label{Vi}
|x_1x_2+x_1x_3+x_2x_3|_p=1, \ \ \ |x_1x_2x_3|_p=|q-1|^3_p<1.
\end{equation}
Since from \eqref{ABC} we have $|A|_p=|B|_p=1$. Now from $x_i=z_i+A/3$, where $z_i$ solutions
of \eqref{eq-alp1}, one gets that
$|x_i|_p\leq 1$. It then follows from \eqref{Vi} that at least one of the solutions' norm
should be
strictly less than 1. Suppose that the norm of two of solutions is strictly less l,
say $|x_1|_p<1$ and
$|x_2|_p<1$. Then from \eqref{Vi} one finds that $|B|_p<1$, which is impossible. So, only
one solution's norm is less than strictly 1, say $x_1$, i.e. $|x_1|_p<1$, $|x_2|_p=|x_3|_p=1$.\\
\end{rem}

Now we are going to investigate solutions of \eqref{eq1} over the
invariant line $(1,1,\dots,h,1,\dots,1)$. Let us introduce some
notations. If $x\in W_n$, then instead of $h_x$ we use the symbol
$h_x^{(n)}$.

Denote
\begin{equation}\label{g(x)}
g_\r(x)=\frac{\r x+q-1}{x+\r+q-2}.
\end{equation}
Note that $f_\r(x)=(g_\r(x))^3$. Then one can see that
\begin{eqnarray}\label{g-xy}
&&|g_\r(x)-g_\r(x)|_p=\frac{|x-y|_p|\r-1|_p|\r+q-1|_p}{|x+\r+q-2|_p|y+\r+q-2|_p},\\[3mm]
\label{g-1} && g^{-1}_\r(x)=\frac{(\r+q-2)x-q+1}{\r-x}
\end{eqnarray}
Moreover, one has the following
\begin{lem}\label{g-p} Let $|q-1|_p<1$, and $|\r|_p\leq|q-1|^2_p$.
The following assertions hold true:
\begin{enumerate}
\item[(i)] If $|x|_p\neq 1$, then $|g_\r(x)|_p\leq\max\{|q-1|_p,|\r|_p\}$;
\item[(ii)] If $|g_\r(x)|_p>1$, then $|x|_p=1$.
\end{enumerate}
\end{lem}
\begin{proof} (i). Let $|x|_p<1$, then from \eqref{g(x)} we get
$$
|g_\r(x)|_p=\bigg|\frac{\r x+q-1}{x+\r+q-2}\bigg|_p=|q-1|_p<1.
$$
Now assume $|x|_p>1$, then analogously one finds
$$
|g_\r(x)|_p\left\{
\begin{array}{ll}
=|\r|_p, \ \ \textrm{if} \ |x|_p>\frac{|q-1|_p}{|\t|_p},\\[2mm]
\leq |q-1|_p, \ \ \textrm{if} \ 1<|x|_p\leq \frac{|q-1|_p}{|\r|_p}.
\end{array}
\right.
$$

(ii) Denoting $y=g_\r(x)$, from \eqref{g-1} one finds
$$
|x|_p=|g^{-1}_\r(y)|_p=\bigg|\frac{(\r+q-2)y-q+1}{\r-y}\bigg|_p=\frac{|y|_p}{|y|_p}=1.
$$
\end{proof}

\begin{thm}\label{U-1} Let $|\r|_p<|q-1|^2_p<1$ and $\sqrt{-3}\in\bq_p$.
Assume that $\{h_x\}_{x\in V\setminus\{(0)\}}$ is a solution of
\eqref{eq1} such that $|h_x|_p\neq 1$ for all $x\in
V\setminus\{(0)\}$. Then $h_x=x_1$ for every $x$.
\end{thm}
\begin{proof} Let us first show that  $|h_x|_p<1$ for all $x$.
Suppose that $|h_x^{(n_0)}|_p>1$ for some $n_0\in\bn$ and $x\in
W_{n_0}$. Since $\{h_x\}$ is a solution of \eqref{eq1}, therefore,
we have
\begin{equation}\label{h-x}
h_x^{(n_0)}=g_\r(h_{(x,1)}^{(n_0+1)})g_\r(h_{(x,2)}^{(n_0+1)})g_\r(h_{(x,3)}^{(n_0+1)}),
\end{equation}
here we have used coordinate structure of the tree.

Now according to $|h_{(x,i)}^{(n_0+1)}|_p\neq 1$, ($i=1,2,3$), then Lemma \ref{g-p} (i) implies
that $|g_\r(h_{(x,i)}^{(n_0+1)})|_p<1$, which with \eqref{h-x} means
$|h_{x}^{(n_0)}|_p<1$. It is a contradiction.

Hence, $|h_x|_p<1$ for all $x$. Then from \eqref{g-xy} we obtain
\begin{equation}\label{g-xy-11}
|g_\r(h_{x})-g_\r(x_1)|_p=|q-1|_p|h_x-x_1|_p
\end{equation}
for any $x\in V\setminus\{(0)\}$.

Now denote
$$
\|h^{(n)}\|_p=\max\{|h_{x}^{(n)}|_p:\ x\in W_n\}.
$$

Let $\e>0$ be an arbitrary number. Then from the proof of Lemma
\ref{g-p} (i) with \eqref{g-xy-11}, Lemma \ref{pr} one finds
\begin{eqnarray}\label{f-hx-x11}
|h_x^{(n)}-x_1|_p&=&\bigg|\prod_{i=1}^3g_\r(h_{(x,i)}^{(n+1)})-(g_\r(x_1))^3\bigg|_p\nonumber\\
&\leq &
|q-1|^2_p\max_{i=1,2,3}\bigg\{|h_{(x,i)}^{(n+1)}-x_1|_p\bigg\}.
\nonumber
\end{eqnarray}
Thus, we derive
$$
\|h^{(n)}-x_1\|_p\leq |q-1|^2_p\|h^{(n+1)}-x_1\|_p.
$$
So, iterating the last inequality $N$ times one gets
\begin{equation}\label{f-iter1}
\|h^{(n)}-x_1\|_p\leq |q-1|^{2^N}_p\|h^{(n+N)}-x_1\|_p.
\end{equation}
Choosing $N$ such that $|q-1|^{2^N}_p<\epsilon$, from \eqref{f-iter1} we
find $\|h^{(n)}-x_1\|_p<\e$. Arbitrariness of $\epsilon$ yields that
$h_x=x_1$. This completes the proof.
\end{proof}

From this theorem we conclude that other solutions of \eqref{eq1}
may exist under condition $|h_x|_p=1$ for all $x\in
V\setminus\{(0)\}$. Note that using arguments of \cite{MR2} one can
show the existence of periodic solutions of \eqref{eq1}.

\section{Boundedness and phase transitions}

In the previous section we have established the existence of generalized $p$-adic quasi
Gibbs measures. In this
section we are going to investigate boundedness of the measures, and
moreover, establish the occurrence of phase transitions.

First we need an auxiliary result. Assume that $\h$ is a solution of \eqref{eq1},
then one finds a constant $a_\h(x)\in\bq_p$ such that
\begin{equation}\label{aN1}
\prod_{y\in
S(x)}\sum_{j=0}^{q-1}\rho^{\d_{ij}}h_{j,y}=a_{\h}(x)h_{i,x}
\end{equation}
for any $i\in\{0,\dots,q-1\}$. This implies
\begin{eqnarray}\label{aN2}
\prod_{x\in W_{n}}\prod_{y\in S(x)}\sum_{j=0}^{q-1}
\rho^{\d_{ij}}h_{j,y}=\prod_{x\in
W_n}a_{\h}(x)h_{i,x}=A_{\h,n}\prod_{x\in W_n}h_{i,x},
\end{eqnarray}
where
\begin{equation}\label{aN3}
A_{\h,n}=\prod_{x\in W_n}a_{\h}(x).
\end{equation}
Given $j\in\Phi$, by $\eta^{(j)}\in\Om_{W_n}$ we denote a
configuration on $W_n$ defined as follows: $\eta^{(j)}(x)=j$ for all
$x\in W_n$.

Hence, by  \eqref{mu},\eqref{aN2} we have
\begin{eqnarray*}
1&=&\sum_{\s\in\Om_n}\sum_{\w\in\Om_{W_n}}\m^{(n+1)}_\h(\s\vee \w)\\
&=&\sum_{\s\in\Om_n}\sum_{\w\in\Om_{W_n}}\frac{1}{Z^{(\h)}_{n+1,\rho}}\rho^{H(\s\vee
\w)}\prod_{x\in
W_{n+1}}h_{\w(x),x}\\
&=&\frac{1}{Z^{(\h)}_{n+1,\rho}}\sum_{\s\in\Om_n}\rho^{H(\s)}
\prod_{x\in W_n}\prod_{y\in S(x)}\sum_{j=0}^q \rho^{\d_{\s(x),j}}h_{j,y}\\
&=&\frac{A_{\h,n}}{Z^{(\h)}_{n+1,\rho}}
\sum_{\s\in\Om_n}\rho^{H(\s)}\prod_{x\in W_n}h_{\s(x),x}\\
&=&\frac{A_{\h,n}}{Z^{(\h)}_{n+1,\rho}}Z_{n,\rho}^{(\h)}
\end{eqnarray*}
which implies
\begin{equation}\label{ZN2}
Z^{(\h)}_{n+1,\rho}=A_{\h,n}Z^{(\h)}_{n,\rho}.
\end{equation}

In this section we basically investigate measures corresponding to solutions of \eqref{eq12}.
Let $\tilde x$ be a solution of \eqref{eq12}, and $\tilde\m$ be the corresponding $p$-adic measure
to such a solution. Then according to \eqref{ZN2} the
partition function $Z^{(\tilde x)}_{\r,n}$ corresponding to the measure $\tilde\m$
has the following form
\begin{equation}\label{Z-IN}
Z^{(\tilde x)}_{\r,n}=a_i^{|V_{n-1}|}
\end{equation}
where $a=(\tilde x+\r+q-2)^3$. Note that when $q\geq 3$, one can check that $a\neq 0$.

For a given configuration $\s\in\Om_{V_n}$ denote
$$
\#\s=\{x\in W_n:\ \s(x)=1\}.
$$

From \eqref{mu},\eqref{H} and \eqref{Z-IN} we find
\begin{eqnarray}\label{e-mui}
|\tilde\m(\s)|_p&=&\frac{|\rho|_p^{H(\s)}}{|Z^{(\tilde x)}_{\r,n}|_p}\cdot
\prod_{x\in W_n}\big|h_{\s(x),x}\big|_p\nonumber\\
&=&\frac{|\tilde x|_p^{\#\s}|\rho|_p^{H(\s)}}{|\tilde x+\r+q-2|_p^{3|V_{n-1}|}}.
\end{eqnarray}

\subsection{Regime $\r\in\ce_p$}

In this subsection we consider two modes: $|q|_p<1$ and $|q|_p=1$.

In this section we shall prove the existence of a phase transition.
Namely one has the following

\begin{thm}\label{bound1}
Let $p>3$, $|q|_p\leq p^{-3}$, $\r\in\ce_p$ with $|\r-1|_p\leq p^{-3}$.
Then for the
$p$-adic $q$-state Potts model \eqref{Potts} on the Cayley tree of order three
the $p$-adic Gibbs measures $\m_k$, $k=0,1,2,3$ are unbounded.
Moreover, there is a phase transition.
\end{thm}

\begin{proof} Note that under the condition of the theorem, there are four
$p$-adic Gibbs measures $\m_k$, $k=0,1,2,3$ (see Theorem \ref{1-e}). In this case, for
the solutions $x_k$ $(k=0,1,2,3)$ of \eqref{eq12} we have $x_k\in\ce_p$. Therefore,
one gets
$$
|x_k+\r+q-2|_p=|x_k-1+\r-1+q|_p\leq\frac{1}{p}.
$$
Hence, from \eqref{e-mui} one finds
\begin{eqnarray*}
|\m_k(\s)|_p=\frac{1}{|x_k+\r+q-2|_p^{3|V_{n-1}|}}\geq p^{3|V_{n-1}|}
\end{eqnarray*}
which implies the unboundedness of $\m_k$. Thus, there is a phase transition.
\end{proof}

Now consider the case $|q|_p=1$. In this setting one has

\begin{thm}\label{bound2}
Let $p>3$,  $\r\in\ce_p$ with $|\r-1|_p\leq |4q-3|_p$, and $|q|_p=|q-1|_p=|2q-1|_p=1$.
Let $\m_0$ be the $p$-adic Gibbs measure and assume that the conditions of Theorem \ref{1-2e}
are satisfied, i.e. there is at least one
translation-invariant generalized $p$-adic quasi Gibbs measure $\tilde\m$ for the
$p$-adic $q$-state Potts model \eqref{Potts} on the Cayley tree of order three.
Then the measure $\m_0$ is bounded, and $\tilde\m$ is unbounded. This means that
there is a phase transition.
\end{thm}

\begin{proof} We recall that the measure $\m_0$ corresponds to the solution $x_0=1$.
It is clear that
$$
|x_0+\r+q-2|_p=|\r-1+q|_p=1.
$$
Hence, from \eqref{e-mui} one finds
\begin{eqnarray*}
|\m_0(\s)|_p=\frac{1}{|x_0+\r+q-2|_p^{3|V_{n-1}|}}=1
\end{eqnarray*}
which yields the boundedness of $\m_0$.

Now
assume that $\tilde\m$ corresponds to a solution $\tilde x$ of \eqref{eq13}. Note that one has
$\tilde x=\tilde z+A/3$, where $\tilde z$ is a solution of \eqref{eq-alp1} such that $|\tilde z|_p<1$.

So, from \eqref{A} we have
\begin{eqnarray}\label{xx}
|\tilde x+\r+q-2|_p&=&\bigg|\tilde z+\frac{A}{3}+\r+q-2\bigg|_p\nonumber\\[2mm]
&=&\bigg|\tilde z+1-(\r-1)+\frac{\r^3-1}{3}-q+\r+q-2\bigg|_p\nonumber\\ [2mm]
&=&|3\tilde z+(\r^3-1)|_p\leq \frac{1}{p}
\end{eqnarray}

Hence, from \eqref{e-mui} with \eqref{xx} one gets
$|\tilde\m(\s)|_p\geq p^{3|V_{n-1}|}$, so $\tilde \m$ is unbounded. This implies that the
existence of a phase transition.
\end{proof}

\subsection{Regime $|\r|_p<1$}

In this subsection we are going to establish the existence of quasi phase transition.

\begin{thm}\label{bound3}
Let $p>3$,  $|\r|_p,|q|_p\leq p^{-2}$. Assume that the conditions of Theorem \ref{1-r} are
satisfied. In this case, there two translation-invariant generalized $p$-adic Gibbs
measures $\m_0$ and $\tilde\m$. Then the measures $\m_0$, $\tilde\m$ are bounded.
This means that
there is a quasi phase transition.
\end{thm}

\begin{proof} Note that the measure $\m_0$ corresponds to the solution $x_0=1$.
So, we have
$$
|x_0+\r+q-2|_p=|\r+q-1|_p=1.
$$
Hence, from \eqref{e-mui} one finds
\begin{eqnarray*}
|\m_0(\s)|_p=\frac{|\rho|_p^{H(\s)}}{|x_0+\r+q-2|_p^{3|V_{n-1}|}}\leq |\r|_p.
\end{eqnarray*}
Therefore, $\m_0$ is bounded.

Now
assume that $\tilde\m$ corresponds to a solution $\tilde x$ of \eqref{eq13}. Note that
one has $\tilde x=\tilde z+A/3$, where $\tilde z$ is a solution of \eqref{eq-alp1}
such that
$$
|\tilde z|_p=
\left\{
\begin{array}{ll}
1, \ \ \textrm{if} \ \ p\neq 19,\\
p^{-1}, \ \ \textrm{if} \ \ p=19.
\end{array}
\right.
$$
for the detail we refer to the proof of Theorem \ref{1-r}.

So, from \eqref{A} we have
\begin{eqnarray}\label{xx1}
|\tilde x+\r+q-2|_p&=&\bigg|\tilde z+\frac{A}{3}+q-1\bigg|_p\nonumber\\[2mm]
&=&\bigg|\tilde z-\frac{1}{3}+\frac{\r^3}{3}\bigg|_p.
\end{eqnarray}

We want to show that $|\tilde z-1/3|_p=1$. If $|\tilde z|_p<1$, then it is evident.

Let $|\tilde z|_p=1$, and assume that $|\tilde z-1/3|_p<1$. Then it yields that
$\tilde z=1/3+\t\eta_1 p$ for some $\eta_1\in\bz_p$. Since, $\tilde z$ is a solution of
\eqref{eq-alp1}, then from \eqref{aabb} one gets
$$
\frac{1}{3^3}+\tilde\eta_1p-\bigg(-\frac{4}{3}+\e p^2\bigg)\bigg(\frac{1}{3}+\eta_1p\bigg)=
-\frac{38}{27}+\e_1p^2.
$$
This implies that
$$
\frac{1}{27}-\frac{4}{9}+\frac{38}{27}\equiv 0 \ (\textrm{mod}\ p).
$$
The last one is equivalent to $1\equiv 0 \ (\textrm{mod}\ p)$ which is impossible. Therefore,
$|\tilde z-1/3|_p=1$, this with \eqref{xx1} yields that
$|\tilde x+\r+q-2|_p=1$. Consequently, from \eqref{e-mui} one gets that the measure
$\tilde\m$ is bounded. Thus, a quasi phase transition occurs.
\end{proof}

Note that when $k=2$ the similar phase transition has occurred for the model (see \cite{Mq}).

\begin{thm}\label{bound3}
Let $p>3$,  $|\r|_p<|q-1|_p^2<1$, and assume that $\sqrt{-3}$ exists. Then for the
translation-invariant generalized $p$-adic quasi Gibbs
measures $\m_k$ ($k=0,1,2,3$) one has: the measures $\m_0$ is unbounded, but
the measures $\m_k$ $(k=1,2,3$) are bounded.
Moreover, there is a strong phase transition.
\end{thm}

\begin{proof} We again recall that the measure $\m_0$ corresponds to the solution $x_0=1$.
So, one has
$$
|x_0+\r+q-2|_p=|\r+q-1|_p=|q-1|_p.
$$
Hence, from \eqref{e-mui} one finds
\begin{eqnarray}\label{e-mu22}
|\m_0(\s)|_p=\frac{|\rho|_p^{H(\s)}}{|x_0+\r+q-2|_p^{3|V_{n-1}|}}
= \frac{|\r|_p^{H(\s)}}{|q-1|_p^{3|V_{n-1}|}}.
\end{eqnarray}

Now let us choose $\s_{0,n}\in\Om_{V_{2n}}$ as follows
\begin{eqnarray*}
\s_{0,n}(x)= \left\{
\begin{array}{ll}
1, \ \ x\in W_{2k}, \\[2mm]
0, \ \ x\in W_{2k-1},
\end{array}
\right. \ \ 1\leq k\leq n.
\end{eqnarray*}
Then one can see that $H(\s_{0,n})=0$, therefore it follows from
\eqref{e-mu22} that
$$
|\m_0(\s_{0,n})|_p\geq p^{2|V_{2n-1}|}\to\infty \ \ \textrm{as} \ \
n\to\infty.
$$
This yields that the measure $\m_0$ is not bounded.

From Remark \ref{41} we know that the equation \eqref{eq13} has three solutions
$x_1,x_2$ and $x_3$ such that $|x_1|_p<1$ and $|x_2|_p=|x_3|_p=1$.
Let $\m_i$ be the corresponding
generalized $p$-adic quasi Gibbs measures ($k=1,2,3$). Note that the solutions $x_i$ are defined
by $z_i$ solutions of \eqref{eq-alp1} via $x_i=z_i+A/3$.

Now consider the measure $\m_1$. From $|x_1|_p<1$ it follows that
$|x_1+\r+q-2|_p=1$, which immediately yields that $\m_1$ is bounded.

To study $\m_{2,3}$ we need some analysis.
From the proof of Theorem \ref{2-r} we know that $|D|_p<1$. Due to \eqref{Dand3x_0^2+a}
one finds that
$$
\bigg|-3\bigg(\frac{z_1}{2}\bigg)^2+\a\bigg|_p=|D|_p<1
$$
Now from \eqref{x+-} we obtain that
\begin{equation}\label{z23}
\bigg|z_{2,3}+\frac{z_1}{2}\bigg|_p<1.
\end{equation}

On the other hand, from the proof of Theorem \ref{numberofsolutionsforp>3} one finds
that $|3z_1^2+\a|_p<1$. The last inequality with \eqref{alpha5} yields that
$$
\bigg|3z_1^2-\frac{1}{3}+\eta p\bigg|_p<1,
$$
where $\eta\in\bz_p$. Then using the last one we have
\begin{eqnarray}\label{z234}
\bigg|z_1-\frac{1}{3}\bigg|_p\bigg|z_1+\frac{1}{3}\bigg|_p&=&
\bigg|z_1^2-\frac{1}{9}\bigg|_p \nonumber  \\[3mm]
&=&\bigg|3z_1^2-\frac{1}{3}\bigg|_p \nonumber \\
&=&\bigg|\bigg(3z_1^2-\frac{1}{3}+\eta p\bigg)-\eta p\bigg|_p<1
\end{eqnarray}

Hence, \eqref{z234} implies that
$$
\bigg|z_1-\frac{1}{3}\bigg|_p<1 \ \ \ \textrm{or} \ \ \  \bigg|z_1-\frac{1}{3}\bigg|_p<1.
$$
In both cases, we infer that
\begin{equation}\label{1z2}
\bigg|\frac{z_1}{2}+\frac{1}{3}\bigg|_p =\bigg|z_1+\frac{2}{3}\bigg|_p=1.
\end{equation}

Now using \eqref{z23} and \eqref{1z2} one gets
\begin{eqnarray}\label{1z3}
|3z_{2,3}-1|_p=\bigg|z_{2,3}-\frac{1}{3}\bigg|_p=
\bigg|\bigg(z_{2,3}+\frac{z_1}{2}\bigg)-\bigg(\frac{z_1}{2}+\frac{1}{3}\bigg)\bigg|_p=1.
\end{eqnarray}

By means of \eqref{1z3} we find
$$
|x_{2,3}+\r+q-2|_p=\bigg|z_{2,3}+A/3+\r+q-2\bigg|_p=
|3z_{2,3}-1+\r^3|_p=1.
$$

Consequently, the last equality with \eqref{e-mui} yields that
\begin{equation}\label{m111}
\m_{2,3}(\s)=|\r|_p^{H(\s)}\leq 1.
\end{equation}
So, the measures $\m_{2,3}$ are bounded.

Now consider the following multiplication (see \eqref{e-mu22},\eqref{m111})
\begin{equation}\label{m112}
|\m_0(\s)\m_2(\s)|_p=\frac{|\r|_p^{2H(\s)}}{|q-1|_p^{3|V_{n-1}|}}
\end{equation}

Let $|\r|_p=p^{-k}$ and $|q-1|_p=p^{-m}$. Then from $|\r|_p<|q-1|_p^2<1$ one finds that
$0<2m<k$. Hence, \eqref{m112} can be rewritten as follows
\begin{equation}\label{m113}
|\m_0(\s)\m_2(\s)|_p=p^{3m|V_{n-1}|-2kH(\s)}.
\end{equation}

Now we choose $\tilde\s_{n}\in\Om_{V_{n}}$ as follows
\begin{eqnarray*}
\tilde\s_{n}(x)= \left\{
\begin{array}{ll}
1, \ \ x\in V_{n-1}, \\[2mm]
0, \ \ x\in W_n,
\end{array}
\right.
\end{eqnarray*}
Then one can see that $H(\tilde\s_{n})=|V_{n-1}|$.

From $3m<4m<2k$ we conclude that
$$
3m|V_{n-1}|<2k|V_{n-1}|=2kH(\tilde\s_n)
$$
so, the last inequality with \eqref{m113} yields that
\begin{equation}\label{m114}
|\m_0(\tilde\s_n)\m_2(\tilde\s_n)|_p\leq 1 \ \ \textrm{for all} \ \ n\geq 1.
\end{equation}

For the sake of simplicity we assume that $3m>k$. Then one finds
$$
3m|V_{n-1}|-kH(\tilde\s_n)=(3m-k)|V_{n-1}|>0
$$
which with \eqref{e-mu22} implies that $|\m_0(\tilde\s_n)|_p\to\infty$ as $n\to\infty$. This
with \eqref{m114} yields that $|\m_2(\tilde\s_n)|_p\to 0$ as $n\to\infty$. Consequently,
there exists a strong phase transition.

Note that if $3m\leq k$, the same phase transition occurs, but it is
technically involved. Therefore, we leave this case.
\end{proof}

\section{Conclusions}

In the present paper, we have studied a phase transition problem
for the $q$-state $p$-adic Potts model over the
Cayley tree of order three. We consider a more general notion
of $p$-adic Gibbs measure which depends on parameter $\rho\in\bq_p$.
Such a measure
is called {\it generalized $p$-adic quasi Gibbs measure}. When $\rho$ equals to $p$-adic
exponent, then it coincides with the usual $p$-adic Gibbs measure (see \cite{MR1}).
When $\rho=p$, then
it coincides with $p$-adic quasi Gibbs measure (see \cite{Mq}). In the present paper
we have considered
two regimes with respect to the values of $|\rho|_p$. Namely, in the first
regime, one takes $\rho=exp_p(J)$ for some $J\in\bq_p$, in the second one we let $|\rho|_p<1$.
In each regime, we first find conditions for the existence of generalized
$p$-adic quasi Gibbs measures. Furthermore, in the first regime, we
established the existence of
the phase transition under the condition $|q|_p\leq p^{-3}$.
If $|q|_p=1$, we
proved the existence of a phase transition under some conditions. In the second regime,
when $|\r|_p,|q|_p\leq p^{-2}$ we proved the existence of a quasi phase transition. It
turns out that if
$|\r|_p<|q-1|_p^2<1$, then one finds the existence of
the strong phase transition.

\section*{Acknowledgement} The first named author (F.M.)
acknowledges the Scientific and Technological Research Council of Turkey (TUBITAK)
for support, and Zirve University for kind hospitality.

\appendix

\section{Cubic equation $x^3+ax=b$ over $\bz_p^*$}

In this section we are going to provide a criterion of the existence
of a solution of the cubic equation
\begin{equation}\label{cubiceqn}
x^3+ax=b,
\end{equation}
where $a,b\in\bq_p$. \footnote{Material of this section is taken
from \cite{MOS}}

Let us consider the following depressed cubic equation in the field
$\bbf_p=\bz/p\bz$
\begin{eqnarray}\label{cubiccong}
x^3+\bar{a}x=\bar{b},
\end{eqnarray}
where $\bar{a},\bar{b}\in \bbf_p.$  We assume that
$\bar{a}\neq\bar{0}$ and $\bar{b}\neq\bar{0}$. The number of
solutions ${\mathbf{N}}_{\bbf_p}(x^3+\bar{a}x-\bar{b})$ of this
equation was described in \cite{ZHS2}.

\begin{prop}[\cite{ZHS2}]\label{CubicinF_p}
Let $p>3$ be a prime number and $\bar{a},\bar{b}\in\bbf_p$ with
$\bar{a}\bar{b}\neq\bar{0}$. Let
$\overline{D}=-4\bar{a}^3-27\bar{b}^2$ and
$u_{n+3}=\bar{b}u_n-\bar{a}u_{n+1}$ for $n\in\bn$ with
$u_1=\bar{0},$ $u_2=-\bar{a},$ $u_3=\bar{b}.$ Then the following
holds true:
$$
{\mathbf{N}}_{\bbf_p}(x^3+\bar{a}x-\bar{b})=\left\{
\begin{array}{l}
3 \ \ \  if \ \ \ \overline{D}u_{p-2}^2=\bar{0} \\
0 \ \ \ if \ \ \ \overline{D}u_{p-2}^2=9\bar{a}^2 \\
1 \ \ \ if \ \ \ \overline{D}u_{p-2}^2\neq \bar{0}, 9\bar{a}^2
\end{array}
\right.
$$
\end{prop}

\begin{prop}[\cite{ZHS2}]\label{numberofcongequation}
Let $p>3$ be a prime number and $\bar{a},\bar{b}\in\bbf_p$,
$\bar{a}\bar{b}\neq\bar{0}$. Let
$\overline{D}=-4\bar{a}^3-27\bar{b}^2$ and
$u_{n+3}=\bar{b}u_n-\bar{a}u_{n+1}$ with
$\overline{D}u_{p-2}^2\neq9\bar{a}^2,$ $u_1=\bar{0},$
$u_2=-\bar{a},$ $u_3=\bar{b}.$
\begin{itemize}
  \item [I] Let $\overline{D}u_{p-2}^2=\bar{0}.$ Then the following statements holds true:
  \begin{itemize}
  \item [I.1] The equation \eqref{cubiccong} has 3 distinct solutions in $\bbf_p$ if and only if $\overline{D}\neq\bar{0}.$ Moreover, one has that $3\bar{x}^2+\bar{a}\neq0$ for any root $\bar{x};$
  \item [I.2] The equation \eqref{cubiccong} has 2 distinct solutions in $\bbf_p$ while one of them of multiplicity 2 if and only if $\bar{D}=\bar{0}.$ If $\bar{x}_1$, $\bar{x}_2$ are 2 distinct solutions while $\bar{x}_1$ is a multiple solution then $\bar{x}_1=\frac{3\bar{b}}{2\bar{a}},$ $\bar{x}_2=-\frac{3\bar{b}}{\bar{a}},$ and $3\bar{x}_2^2+\bar{a}\neq\bar{0};$
  \item [I.3] The equation \eqref{cubiccong} does not have any solution of multiplicity 3.
\end{itemize}

  \item [II] Let $\overline{D}u_{p-2}^2\neq\bar{0}, 9\bar{a}^2.$ If $\bar{x}$ is a solution of the equation \eqref{cubiccong} then $3\bar{x}^2+\bar{a}\neq\bar{0}.$
\end{itemize}
\end{prop}

\begin{rem} Due to Proposition \ref{numberofcongequation}, one may conclude that under the assumption of Proposition \ref{numberofcongequation}, there always exists at least one solution $\bar{x}$ of the equation \eqref{cubiccong} such that $3\bar{x}^2+\bar{a}\neq\bar{0}.$
\end{rem}

In what follows we need the following auxiliary result.

\begin{prop}\label{NecessaryconditionforZ_p^*}
Let $p$ be any prime. Suppose that the equation \eqref{cubiceqn} has
a solution in $\bz_p^{*},$  where $a,b\in\bq_p$ with $ab\neq 0.$
Then one of the following conditions holds true:
\begin{itemize}
  \item [(i)] $|a|_p<|b|_p=1;$
  \item [(ii)] $|b|_p<|a|_p=1;$
  \item [(iii)] $|a|_p=|b|_p\geq 1.$
\end{itemize}
\end{prop}

The proof is straightforward.

We are going to state the solvability criterion for the depressed
cubic equation \eqref{cubiceqn} in domain $\bz_p^{*}$ at $p>3.$

Let $a,b\in\bq_p$ be two nonzero $p-$adic numbers with
$a=\frac{a^{*}}{|a|_p},$ $b=\frac{b^{*}}{|b|_p}$ where
$a^{*},b^{*}\in\bz_p^{*}$ with $a^{*}=a_0+a_1\cdot p+a_2\cdot
p^2+\cdots$ and $b^{*}=b_0+b_1\cdot p+b_2\cdot p^2+\cdots.$

We set $D_0=-4a_0^3-27b_0^2$ and $u_{n+3}=b_0u_n-a_0u_{n+1}$ with
$u_1=0,$ $u_2=-a_0,$ and $u_3=b_0$ for $n=\overline{1,p-3}$

\begin{thm}\label{criterionforp>3}
Let $p>3$ be a prime. Then the equation \eqref{cubiceqn} has a
solution in $\bz_p^{*}$ if and only if one of the following
conditions hold true:
  \begin{itemize}
    \item [I.1] $|a|_p<|b|_p=1$ and $b_0^{\frac{p-1}{(3,p-1)}}\equiv 1 \ (mod \ p);$
    \item [I.2] $|b|_p<|a|_p=1$ and $(-a_0)^{\frac{p-1}{2}}\equiv 1 \ (mod \ p);$
    \item [I.3] $|a|_p=|b|_p=1$ and $D_0u_{p-2}^2\not\equiv 9a_0^{2} \ (mod \ p);$
    \item [I.4] $|a|_p=|b|_p>1.$
  \end{itemize}
\end{thm}

\begin{proof}
Note that if  the equation \eqref{cubiceqn}
has a solution in $\bz_p^{*}$, then due to Proposition
\ref{NecessaryconditionforZ_p^*} there are three possible cases.
Let us consider every case.

I.1.  Assume that the equation \eqref{cubiceqn}
has a solution $x\in\bz_p^{*}$, when $|a|_p<1,$
$|b|_p=1$.  Due to $|a|_p<1$ and $|b|_p=1,$ we get
$$
x_0^3+ax_0\equiv x_0^3 \equiv b_0 \ (mod \ p).
$$
Hence Proposition \ref{aisresidueofp} yields that
$b_0^{\frac{p-1}{(3,p-1)}}\equiv 1 \ (mod \ p).$

Now assume that  $|a|_p<1,$  $|b|_p=1,$ and
$b_0^{\frac{p-1}{(3,p-1)}}\equiv 1 \ (mod \ p)$ are satisfied. From
$b_0^{\frac{p-1}{(3,p-1)}}\equiv 1 \ (mod \ p)$ due to Proposition \ref{aisresidueofp} one
can find $x_0\in\bz$
such that $x_0^3\equiv b_0 \ (mod \ p)$ and $(x_0,p)=1$. Let us
consider the following polynomial $f_{a,b}(x)=x^3+ax-b.$ By $|a|_p<1$ one gets
$$
f_{a,b}(x_0)=x_0^3+ax_0-b\equiv x_0^3-b_0\equiv 0 \ (mod \ p),\quad
f'_{a,b}(x_0)=3x_0^2+a\not\equiv 0 \ (mod \  p).
$$
Hence, Hensel's Lemma \ref{Hensel} implies the existence of a solution $x\in\bz_p$ of \eqref{cubiceqn}
with $|x-x_0|_p\leq \frac{1}{p}$, where
$|x_0|_p=1.$ This means that $|x|_p=1$, i.e., $x\in\bz_p^{*}.$\\

I.2. Assume that the equation \eqref{cubiceqn} has a
solution $x\in\bz_p^{*}$ when $|a|_p=1,$
$|b|_p<1$.  Then it follows that
$$
x_0^3+a_0x_0\equiv x_0(x_0^2+a_0)\equiv b\equiv 0 \ (mod \ p).
$$
This means that $x_0^2\equiv -a_0 \ (mod \ p)$, and again
Proposition \ref{aisresidueofp} implies
$(-a_0)^{\frac{p-1}{2}}\equiv 1 \ (mod \ p).$

Now we suppose that $|a|_p=1,$  $|b|_p<1,$ and
$(-a_0)^{\frac{p-1}{2}}\equiv 1 \ (mod \ p)$.  The congruence
$(-a_0)^{\frac{p-1}{2}}\equiv 1 \ (mod \ p)$ yields the existence $x_0\in\bz$
such that $x_0^2+a_0\equiv 0 \ (mod \ p)$ and $(x_0,p)=1$.
Again consider the polynomial $f_{a,b}(x)$. Due to $|b|_p<1$ one finds
\begin{eqnarray*}
f_{a,b}(x_0)&=&x_0^3+ax_0-b\equiv x_0(x_0^2+a_0)\equiv 0 \ (mod \ p),\\
f'_{a,b}(x_0)&=&3x_0^2+a\equiv3(x_0^2+a_0)-2a_0\not\equiv 0 \ (mod \  p).
\end{eqnarray*}
Hence, Hensel's Lemma \ref{Hensel} implies that the equation \eqref{cubiceqn} has
a solution $x\in\bz_p$ with $|x-x_0|_p\leq \frac{1}{p}$ where
$|x_0|_p=1.$ This means $x\in\bz_p^{*}.$

I.3. Assume that the equation \eqref{cubiceqn} has a
solution $x\in\bz_p^{*}$ when $|a|_p=1,$
$|b|_p=1$. Then we have
$$
x_0^3+a_0x_0\equiv b_0 \ (mod \ p).
$$
This means that the depressed cubic equation $x^3+a_0x=b_0$ has at
least one solution in the finite field $\bbf_p$. According to
Proposition \ref{CubicinF_p} one has $D_0u_{p-2}^2\not\equiv
9a_0^{2} \ (mod \ p).$

Now we assume that $|a|_p=1,$  $|b|_p=1,$ and
$D_0u_{p-2}^2\not\equiv 9a_0^{2} \ (mod \ p)$ are satisfied. The congruence
$D_0u_{p-2}^2\not\equiv 9a_0^{2} \ (mod \ p)$ yields  the equation
$x^3+a_0x\equiv b_0 \ (mod \ p)$ has at least one solution.
Proposition \ref{numberofcongequation} implies that among all solutions of
$x^3+a_0x\equiv b_0 \ (mod \ p)$, there always exists at
least one solution $x_0$ such that $3x_0^2+a_0\not\equiv 0 \ (mod \
p)$ with $(x_0,p)=1.$

Again, if we apply Hensel's Lemma \ref{Hensel} to $f_{a,b}(x)$
at the point $x=x_0$ then we get that \eqref{cubiceqn} has
a solution $x\in\bz_p$ with $|x-x_0|_p\leq
\frac{1}{p}$ where $|x_0|_p=1.$ So, $x\in\bz_p^{*}.$

I.4. In this case, it is enough to establish the existence of a solution of
\eqref{cubiceqn} under the condition $|a|_p=|b|_p>1$. We have  $a=p^{-m}a^{*},$
$b=p^{-m}b^{*}$ with $a^{*},b^{*}\in\bz_p^{*}$, for some $m\geq 1$. It is clear that the
depressed equation $x^3+p^{-m}a^{*}x=p^{-m}b^{*}$ has a solution in
$\bz_p^{*}$ if and only if the equation $p^{m}x^3+a^{*}x=b^{*}$ has
a solution in $\bz_p^{*}$, where $a^{*},b^{*}\in\bz_p^{*}.$ To this
end, we consider the polynomial
$g(x)=p^{m}x^3+a^{*}x-b^{*}.$ If $a_0,b_0$ are the first digits of
$a^{*},b^{*}\in\bz_p^{*}$ then one finds is $x_0\in\bz$ such that
$a_0x_0\equiv b_0 \ (mod \ p)$ with $(x_0,p)=1.$ Then, we obtain
$$
g(x_0)\equiv p^{m}x_0^3+a_0x_0-b_0\equiv 0 \ (mod \ p), \quad
g'(x_0)\equiv 3p^{m}x_0^2+a_0\not\equiv 0\ (mod \ p).
$$
Again, due to Hensel's Lemma \ref{Hensel}, the equation
$p^{m}x^3+a^{*}x=b^{*}$ has a solution $x\in\bz_p^{*}$ with
$|x-x_0|_p\leq \frac{1}{p}$. This means that
$x\in\bz_p^{*}.$
\end{proof}

Let $D=-4(a|a|_p)^3-27(b|b|_p)^2$. We have $D=\frac{D^{*}}{|D|_p}$
whenever $D\neq 0,$ where $D^{*}\in\bz_p^{*}$ and
$D^{*}=d_0+d_1\cdot p+d_2\cdot p^2+\cdots$

\begin{thm}\label{numberofsolutionsforp>3} Let $p>3$ be a prime.
Then one has
$$
{\mathbf{N}}_{\bz_p^{*}}(x^3+ax-b)=\left\{
\begin{array}{l}
3, \ \ \ |a|_p<|b|_p=1, \ p\equiv 1\ (mod \ 3), \ b_0^{\frac{p-1}{3}}\equiv 1\ (mod \ p)   \\
3, \ \ \ |a|_p=|b|_p=1, \ D=0 \\
3, \ \ \ |a|_p=|b|_p=1, \ 0<|D|_p<1, \ 2\mid\log_p|D|_p, d_0^{\frac{p-1}{2}}\equiv1(mod \ p) \\
3, \ \ \ |a|_p=|b|_p=1, \ |D|_p=1, \ u_{p-2}\equiv 0\ (mod \ p)\\
2, \ \ \ |b|_p<|a|_p=1, \ (-a_0)^{\frac{p-1}{2}}\equiv 1 \ (mod \ p)\\
1, \ \ \ |a|_p<|b|_p=1, \ p\equiv 2\ (mod \ 3) \\
1, \ \ \ |a|_p=|b|_p=1, \ 0<|D|_p<1, \ 2\mid ord_p(D),\ d_0^{\frac{p-1}{2}}\not\equiv1(mod \ p) \\
1, \ \ \ |a|_p=|b|_p=1, \ 0<|D|_p<1, \ 2\nmid ord_p(D), \\
1, \ \ \ |a|_p=|b|_p=1, \ D_0u_{p-2}^2\not\equiv 0,9a_0^{2}\ (mod \ p)\\
1, \ \ \ |a|_p=|b|_p>1\\
0, \ \ \ otherwise
\end{array}
\right.
$$
\end{thm}

\begin{proof}
Let $p>3$ be a prime number and $a,b\in\bq_p$ be two nonzero
$p-$adic numbers.

We want to describe the number
${\mathbf{N}}_{\bz_p^{*}}(x^3+ax-b)$ of solutions of the depressed
cubic equation \eqref{cubiceqn} in the domain $\bz_p^{*}.$

According to Theorem \ref{criterionforp>3}, the equation
\eqref{cubiceqn} has a solution in $\bz_p^{*}$ if and
only if one of the conditions I.1-I.4 holds true. We want to find a number
of solutions in every case.

First assume that $|a|_p<|b|_p=1$ and
$b_0^{\frac{p-1}{(3,p-1)}}\equiv 1 \ (mod \ p).$ In this case, we
showed in the proof of I.1, the
number of solutions of \eqref{cubiceqn} in $\bz_p^{*}$
is the same as the number of solutions of the equation $x_0^3= b_0$
in $\bbf_p.$ Hence, Proposition \ref{aisresidueofp} implies that the last
equation has 3 distinct solutions if $p\equiv 1 \ (mod \ 3)$ and it
has a unique solution if $p\equiv 2 \ (mod \ 3).$

Therefore, if $|a|_p<|b|_p=1$, $p\equiv 1 \ (mod \ 3)$, and
$b_0^{\frac{p-1}{3}}\equiv 1 \ (mod \ p)$ then ${\mathbf{N}}_{\bz_p^{*}}(x^3+ax-b)=3$,
and if
$|a|_p<|b|_p=1,$ $p\equiv 2 \ (mod \ 3)$ and $b_0^{p-1}\equiv 1 \
(mod \ p)$ then ${\mathbf{N}}_{\bz_p^{*}}(x^3+ax-b)=1$.
It is worth mentioning that one always has
$b_0^{p-1}\equiv 1 \ (mod \ p)$ since $(b_0,p)=1$.\\

Now consider $|b|_p<|a|_p=1$ and
$(-a_0)^{\frac{p-1}{2}}\equiv 1 \ (mod \ p).$ We showed that the
number of solutions of \eqref{cubiceqn} in $\bz_p^{*}$
is the same as the number of solutions of the equation $x_0^2=-a_0$
in $\bbf_p.$ From  $(-a_0)^{\frac{p-1}{2}}\equiv 1 \ (mod \ p),$ one concludes that the
last equation has 2 distinct solutions in $\bbf_p.$
Hence, ${\mathbf{N}}_{\bz_p^{*}}(x^3+ax-b)=2$.

Now assume that $|a|_p=|b|_p=1$ and
$D_0u_{p-2}^2\not\equiv 9a_0^{2} \ (mod \ p).$ In this case, we know that
the equation \eqref{cubiceqn} has a
solution $\bar{x}$ such that $\bar{x}\equiv\bar{x}_0 \ (mod \ p),$
where $\bar{x}_0$ is a solution of
\begin{eqnarray}\label{I3}
x_0^3+a_0x_0\equiv b_0 \ (mod \ p)
\end{eqnarray}
such that $3\bar{x}_0^2+a_0\not\equiv 0 \ (mod \ p).$

Due to Proposition \ref{CubicinF_p}, if $D_0u_{p-2}^2\not\equiv 0,
9a_0^{2} \ (mod \ p)$ then \eqref{I3} does not have any
solution except $\bar{x}_0$. Therefore, due to Hensel's Lemma
\ref{Hensel}, if $|a|_p=|b|_p=1$ and $D_0u_{p-2}^2\not\equiv 0,
9a_0^{2} \ (mod \ p)$ then ${\mathbf{N}}_{\bz_p^{*}}(x^3+ax-b)=1$.

Now, assume that $|a|_p=|b|_p=1$ and $D_0u_{p-2}^2\equiv 0
\ (mod \ p)$. Then Proposition \ref{CubicinF_p} yields that the equation
\eqref{I3} has 2 more solutions besides
$\bar{x}_0.$ We denote them by $x_0$ and $y_0$. There is no loss of
generality in assuming that $0<x_0,y_0<p.$

If $D_0\not\equiv 0 \
(mod \ p)$ then Proposition \ref{numberofcongequation} implies that all
solutions $x_0,y_0,\bar{x}_0$ of \eqref{I3} are distinct from each other and
one has
$3{x}_0^2+a_0\not\equiv 0 \ (mod \ p),$ $3{y}_0^2+a_0\not\equiv 0 \
(mod \ p).$ Therefore, by applying Hensel's Lemma \ref{Hensel} to
\eqref{cubiceqn} one finds 2 more solutions $x,$ $y$. Moreover, we have
$\bar{x}$ such that ${x}\equiv{x}_0 \ (mod \ p),$ ${y}\equiv{y}_0 \
(mod \ p).$ Hence, $|a|_p=|b|_p=1,$  $D_0\not\equiv 0 \ (mod \ p)$
(equivalently $|D|_p=1$) and $u_{p-2}\equiv 0 \ (mod \ p)$ then ${\mathbf{N}}_{\bz_p^{*}}(x^3+ax-b)=3$.

Now we suppose that $D_0\equiv 0 \ (mod \ p)$. Then
Proposition \ref{numberofcongequation} implies that two solutions $x_0,y_0$ of
\eqref{I3} are equal to each other, and one has
$\bar{x}_0=-2x_0,$ $3{x}_0^2+a_0\equiv 0 \ (mod \ p).$
Let us study this case in a detail.

Since $\bar{x}$ is a solution of the equation \eqref{cubiceqn}, then one
finds
$$
x^3+ax-b=(x-\bar{x})(x^2+\bar{x}x+\bar{x}^2+a).
$$

We are going to study the quadratic equation
\begin{eqnarray}\label{quadraticequat}
x^2+\bar{x}x+\bar{x}^2+a=0.
\end{eqnarray}
It is clear that
$$
\left(x+\frac{\bar{x}}{2}\right)^2=-\left(3\left(\frac{\bar{x}}{2}\right)^2+a\right).
$$
Since $a,\bar{x}\in\bz_p^{*}$ and $p>3$, we have that
$\left|3\left(\frac{\bar{x}}{2}\right)^2+a\right|_p\leq1.$ Since
$\bar{x}$ is a solution of \eqref{cubiceqn}, one gets $
4\bar{x}\left(3\left(\frac{\bar{x}}{2}\right)^2+a\right)=a\bar{x}+3b.
$ We then have that
\begin{eqnarray*}
D&=&-4a^3-27b^2=3(a^2\bar{x}^2-9b^2)-a^2(3\bar{x}^2+4a)\\
&=&12\bar{x}\left(3\left(\frac{\bar{x}}{2}\right)^2+a\right)(a\bar{x}-3b)-4a^2\left(3\left(\frac{\bar{x}}{2}\right)^2+a\right)\\
&=&4\left(3\left(\frac{\bar{x}}{2}\right)^2+a\right)\left(3a\bar{x}^2-9b\bar{x}-a^2\right).
\end{eqnarray*}

Due to $\bar{x}\equiv\bar{x}_0=-2x_0 \ (mod \ p),$ $3x_0^2+a_0\equiv
0 \ (mod \ p)$, and $2x_0a_0\equiv 3b_0\ (mod \ p)$ we find that
$3a\bar{x}^2-9b\bar{x}-a^2\equiv -9a_0^2\ (mod \ p).$ This means
that $\left|3a\bar{x}^2-9b\bar{x}-a^2\right|_p=1$, and
$-\left(3a\bar{x}^2-9b\bar{x}-a^2\right)$ is a complete square of
some $p-$adic integer number. Therefore, we obtain that
\begin{eqnarray}\label{Dand3x_0^2+a}
-\left(3\left(\frac{\bar{x}}{2}\right)^2+a\right)=\frac{D}{-4\left(3a\bar{x}^2-9b\bar{x}-a^2\right)}
\end{eqnarray}

Let us analyze the quadratic equation \eqref{quadraticequat}.

If $D=0$ then \eqref{quadraticequat} has
solutions $x_1=x_2=\frac{3b}{2a}$ in $\bz_p^{*}.$ Therefore,
$|a|_p=|b|_p=1$ and $D=0$, then \eqref{cubiceqn}
has 3 solutions such that $x_1=x_2=\frac{3b}{2a}$ and
$x_3=-\frac{3b}{a.}$

Let $D\neq 0.$ Since $D\equiv D_0\equiv 0 \ (mod \ p)$, there exists
$k\in\bn$ such that $|D|_p=p^{-k}$, i.e., $D=\frac{D^{*}}{|D|_p},$
where $D^{*}\in\bz_p^{*}$ with $D^{*}=d_0+d_1\cdot p+d_2\cdot
p^2+\cdots$.

The quadratic equation \eqref{quadraticequat} has a solution if and
only if $ord_p(D)=-k$ is even number and
$d_0^{\frac{p-1}{2}}\equiv 1 \ (mod \ p)$. In this case
$-\left(3\left(\frac{\bar{x}}{2}\right)^2+a\right)$ is a complete
square and \eqref{quadraticequat} has  two
distinct solutions in $\bz_p^{*}$, which are given by
\begin{eqnarray}\label{x+-}
x_{\pm}=-\frac{\bar{x}}{2}\pm\sqrt{-\left(3\left(\frac{\bar{x}}{2}\right)^2+a\right)}
\end{eqnarray}

So, if $|a|_p=|b|_p=1,$ $0<|D|_p<1,$ $2\mid ord_p(D)$, and
$d_0^{\frac{p-1}{2}}\equiv 1 \ (mod \ p)$ then ${\mathbf{N}}_{\bz_p^{*}}(x^3+ax-b)=3$.
If
$|a|_p=|b|_p=1,$ $0<|D|_p<1,$ $2\mid ord_p(D)$, and
$d_0^{\frac{p-1}{2}}\not\equiv 1 \ (mod \ p)$ or $|a|_p=|b|_p=1,$
$0<|D|_p<1,$ and $2\nmid ord_p(D)$ then
${\mathbf{N}}_{\bz_p^{*}}(x^3+ax-b)=1$.\\

Now consider the case $|a|_p=|b|_p>1.$ We know that
${\mathbf{N}}_{\bz_p^{*}}(x^3+ax-b)$ is the same as
the number of solutions of the linear
equation $a_0x_0=b_0$ in $\bbf_p.$ Since $a_0\neq 0$, the last
equation has a unique solution. Therefore, if  $|a|_p=|b|_p>1$ then
${\mathbf{N}}_{\bz_p^{*}}(x^3+ax-b)=1$.

\end{proof}


\begin{thebibliography}{99}
\bibitem{F1} We
point out that stochastic processes on the field $\bq_p$ of $p$-adic
numbers with values of real numbers have been studied by many
authors, for example, \cite{AK,AZ1,AZ2,DF,Koc,Y}. In those
investigations wide classes of Markov processes on $\bq_p$ were
constructed and studied. In our case the situation is different,
since probability measures take their values in $\bq_p$. This leads
our investigation to some difficulties. For example, there is no
information about the compactness of $p$-adic values probability
measures.
\bibitem{F31} In \cite{KK1,KK2} a notion of ultrametric Markovianity, which
describes independence of contributions to random field from
different ultrametric balls, has been introduced, and showed that
Gaussian random fields on general ultrametric spaces (which were
related with hierarchical trees), which were defined as a solution
of pseudodifferential stochastic equation  \cite{AKS,KaKo},
satisfies the Markovianity. In addition, covariation of the defined
random field was computed with the help of wavelet analysis on
ultrametric spaces (see also \cite{AKS,Koz}). Some applications of
the results to replica matrices, related to general ultrametric
spaces have been investigated in \cite{KK3}.
\bibitem{F4} In the real case, when the
state space is compact, then the existence follows from the
compactness of the set of all probability measures (i.e. Prohorov's
Theorem). When the state space is non-compact, then there is a
Dobrushin's Theorem \cite{Dob1,Dob2} which gives a sufficient
condition for the existence of the Gibbs measure for a large class
of Hamiltonians.
\end{thebibliography}
\end{document}